%% file: main.tex
\documentclass[11pt]{article}

\usepackage{fullpage}
\usepackage[utf8]{inputenc}
\usepackage{amsmath,amssymb}
\usepackage{bbm}
\usepackage{algorithmic}
\usepackage{algorithm}
\usepackage{color}
\usepackage[colorlinks,urlcolor=blue,citecolor=blue,linkcolor=blue]{hyperref}
\usepackage{array}
\usepackage{paralist}
\usepackage{graphicx}
\usepackage{varwidth}

 \def\icalpsub{0}

\begin{document}

\makeatletter
\newcommand{\HEADER}[1]{\ALC@it\underline{\textsc{#1}}\begin{ALC@g}}
\newcommand{\ENDHEADER}{\end{ALC@g}}
\makeatother

\date{}

\newtheorem{theorem}{Theorem}[section]
\newtheorem{lemma}{Lemma}[section]
\newtheorem{claim}[lemma]{Claim}
\newtheorem{definition}{Definition}[section]
\newtheorem{observation}[lemma]{Observation}
\numberwithin{equation}{section}

\def\FullBox{\hbox{\vrule width 8pt height 8pt depth 0pt}}
\newcommand{\qed}{\;\;\;\FullBox}
\newenvironment{proof}{\noindent{\bf Proof:~~}}{\(\qed\)}
\newenvironment{proofof}[1]{\noindent{\bf Proof of {#1}:~~}}{\(\qed\)}

\newcommand{\eqdef}{\stackrel{\rm def}{=}}
\newcommand{\bitset}{\{0,1\}}
\newcommand{\bitsetn}{\{0,1\}^n}
\renewcommand{\th}{{\rm th}}
\renewcommand{\Pr}{{\rm Pr}}
\newcommand{\poly}{{\rm poly}}
\newcommand{\ov}{\overline}
\newcommand{\eps}{\epsilon}
\newcommand{\wt}{{\rm wt}}
\newcommand{\tildeO}{{\tilde{O}}}
\newcommand{\E}{\mathbb{E}} 
\newcommand{\nm}{N}
\newcommand{\nmr}{N} 
\newcommand{\rd}{M}
\newcommand{\rdr}{R}
\newcommand{\nms}{\mathcal{N}} 
\newcommand{\hatnms}{\widehat{\mathcal{N}}}
\newcommand{\calP}{\mathcal{P}}
\newcommand{\calW}{\mathcal{W}}
\newcommand{\calT}{\mathcal{T}}

\newcommand{\Dist}{\Delta}
\newcommand{\wDist}{\widehat{\Dist}}

\newcommand{\calC}{\mathcal{C}}
\newcommand{\calH}{\mathcal{H}}
\newcommand{\calHsin}{\mathcal{H}_{sin}}
\newcommand{\calHmed}{\mathcal{H}_{med}}
\newcommand{\calHsml}{\mathcal{H}_{sml}}

\newcommand{\kd}{k_d}
\newcommand{\kc}{k_c}

\newcommand{\tilC}{\widetilde{C}}
\newcommand{\tilcalC}{\widetilde{\calC}}
\newcommand{\tilT}{\widetilde{T}}
\newcommand{\tilw}{\widetilde{w}}
\newcommand{\tilp}{\widetilde{p}}
\newcommand{\tilk}{\widetilde{k}}
\newcommand{\tiln}{\widetilde{n}}
\newcommand{\tilDist}{\widetilde{\Dist}}

\definecolor{forestgreen}{rgb}{0.13, 0.55, 0.13}

\newcommand{\dchange}[1]{{\color{forestgreen}  #1}}
\newcommand{\dcom}[1]{{\color{forestgreen}[{\bf D: #1}]}}
\newcommand{\ochange}[1]{{\color{blue}  #1}}
\newcommand{\ocom}[1]{{\color{blue}[{\bf O: #1}]}}
\newcommand\Set[2]{\{\,#1\mid#2\,\}}
\newcommand\SET[2]{\Set{#1}{\text{#2}}}

\title{Sample-based distance-approximation for subsequence-freeness}
\author{Omer Cohen Sidon\thanks{Tel Aviv University, omercs123@gmail.com} \and Dana Ron\thanks{Tel Aviv University, danaron@tau.ac.il. Supported by the Israel Science Foundation (grant number~1146/18) and the Kadar-family award.}}

\begin{titlepage}
\maketitle

\begin{abstract}
In this work, we study the problem of approximating the distance to subsequence-freeness in the sample-based distribution-free model.
For a given subsequence (word) $w = w_1 \dots w_k$, a sequence (text) $T = t_1 \dots t_n$
 is said to contain $w$ if there exist indices $1 \leq i_1 < \dots < i_k \leq n$ such that
 $t_{i_{j}} = w_j$ for every $1 \leq j \leq k$.
 Otherwise, $T$ is $w$-free. Ron and Rosin (ACM TOCT 2022) showed that the number of samples both necessary and sufficient for one-sided error testing of subsequence-freeness in the sample-based distribution-free model is $\Theta(k/\eps)$.

 Denoting by $\Dist(T,w,p)$ the distance of $T$ to $w$-freeness under a distribution $p :[n]\to [0,1]$, we are interested in obtaining an estimate $\wDist$, such that  $|\wDist - \Dist(T,w,p)| \leq \delta$
 with probability at least $2/3$,
 for a given distance parameter $\delta$.
 Our main result is an algorithm whose sample complexity is $\tilde{O}(k^2/\delta^2)$.
We first present an algorithm that works when the underlying distribution $p$ is uniform, and then show how it can be modified to work for any (unknown) distribution $p$.
We also show that a quadratic dependence on $1/\delta$ is necessary.
\end{abstract}
\thispagestyle{empty}
\end{titlepage}

\input{intro}


\input{uniform-pref}


\input{dist-free-pref}



\input{lower-bound}


\input{references.bbl}

\appendix

\input{appendix}

\end{document}

%% file: intro.tex
\section{Introduction}
Distance approximation algorithms, as defined in~\cite{PRR}, are sublinear algorithms that approximate (with constant success probability) the distance of objects from satisfying a prespecified property $\calP$.
Distance approximation (and the closely related notion of tolerant testing) is an extension of property testing~\cite{RS,GGR}, where the goal is to distinguish between objects that satisfy a  property $\calP$ and those that are far from satisfying the property.\footnote{Tolerant testing algorithms are required to distinguish between objects that are close to satisfying a property and those that are far from satisfying it.} 
In this work we consider the property of subsequence-freeness.
For a given subsequence (word) $w_1 \dots w_k$
over some alphabet $\Sigma$, a sequence (text) $T = t_1 \dots t_n$ over $\Sigma$ 
is said to be $w$-free
if there do not exist indices $1 \leq j_1 < \dots < j_k \leq n$ such that $t_{j_i} = w_i$ for every $i \in [k]$.\footnote{For an integer $x$, we use $[x]$ to denote the set of integers $\{1,\dots,x\}$}

In most previous works on property testing and distance approximation, the algorithm is allowed query access to the object, and distance to satisfying the property in question, $\calP$, is defined as the minimum Hamming distance to an object that satisfies $\calP$, normalized by the size of the object.
In this work we consider the more challenging, and sometimes more suitable, sample-based model in which the algorithm  is only given a random sample from the object. In particular, when the object is a sequence $T = t_1 \dots t_n$, each element in the sample is a pair $(j,t_j)$.

We study both the case in which
the underlying distribution according to which each index $j$ is selected (independently) is the uniform distribution over $[n]$,
and the more general case in which the underlying distribution is some arbitrary unknown $p :[n] \to [0,1]$.
We refer to the former as the \textit{uniform sample-based model}, and to the latter  as the \textit{distribution-free sample-based model}. 
The distance (to satisfying the property) is determined by the underlying distribution. Namely, it is the minimum total weight according to $p$ of indices $j$ such that $t_j$ must be modified so as to make the sequence  $w$-free.
Hence, in the uniform sample-based model, the distance measure is simply the Hamming distance normalized by $n$. 

The related problem of testing the property of subsequence-freeness in the distribution-free sample-based model was studied by Ron and Rosin~\cite{RR-toct}. They showed that the
sample-complexity of one-sided error testing of subsequence-freeness in this model is $\Theta(k/\eps)$ (where $\epsilon$ is the given distance parameter).
A natural question is whether we can design a sublinear algorithm, with small sample complexity, that actually approximates the distance of a text $T$ to $w$-freeness.
It is worth noting that, in general, tolerant testing (and hence distance-approximation) for a property may be much harder than  testing the property~\cite{FF,BFLR}.


\vspace{-1ex}
\subsection{Our results}
In what follows, when we say that a sample is selected uniformly from $T$, we mean that for each sample point $(j,t_j)$, $j$ is selected uniformly and independently from $[n]$. 
This generalizes to the case in which the underlying distribution is an arbitrary distribution $p$.

We start by designing a distance-approximation algorithm in the uniform sample-based model.
Let $\Dist(T,w)$ denote the distance under the uniform distribution of $T$ from being $w$-free (which equals the fraction of symbols in $T$ that must be modified so as to obtain a $w$-free text),
and let $\delta \in (0,1)$ denote the error parameter given to the algorithm.

\newcommand\UniAlgThmStatement{
\sloppy
There exists a sample-based distance-approximation algorithm for subsequence-freeness under the uniform distribution, that
takes a sample of size
$\Theta\left(\frac{k^2}{\delta^2}\cdot \log\left(\frac{k}{\delta}\right)\right)$ and outputs an estimate $\wDist$
such that $|\wDist - \Dist(T,w)| \leq \delta$ with probability at least $2/3$.\footnote{As usual, we can increase the success probability to $1-\eta$, for any $\eta>0$ at a multiplicative cost of $O(\log(1/\eta))$ in the sample complexity.}
}
\begin{theorem} \label{thm:uni_alg} 
\UniAlgThmStatement
\end{theorem}

We then turn to extending this result to the distribution-free sample-based model. For a distribution $p:[n] \to [0,1]$, 
we use $\Dist(T,w,p)$ to denote the distance of $T$ from $w$-freeness under the distribution $p$ (i.e., the minimum weight, according to $p$, of the symbols in $T$ that must be modified so as to obtain a $w$-free text).
\newcommand\DistFreeThmStatement{
\sloppy
There exists a sample-based distribution-free distance-approximation algorithm for subsequence-freeness, that
takes a sample of size
$\Theta\left(\frac{k^2}{\delta^2}\cdot \log\left(\frac{k}{\delta}\right)\right)$ from $T$, distributed according to an unknown distribution $p$, and outputs an estimate $\wDist$
such that $|\wDist - \Dist(T,w,p)| \leq \delta$ with probability at least $\frac{2}{3}$.
}

\begin{theorem}\label{thm:main}
\DistFreeThmStatement
\end{theorem}

Finally, we address the question of how tight is our upper bound.
We show (using a fairly simple argument) that the quadratic dependence on $1/\delta$ is indeed necessary, even for the uniform distribution. To be precise, denoting by $\kd$  the number of distinct symbols in $w$, we give a lower bound of $\Omega(1/(\kd \delta^2))$ under the uniform distribution (that holds for every $w$ with $\kd$ distinct symbols, sufficiently large $n$ and sufficiently small $\delta$ -- for a precise statement, see Theorem~\ref{thm:lb}).
\vspace{-1ex}
\subsection{A high-level discussion of our algorithms}
\label{subsec:ideas}
Our starting point is a structural characterization of the distance to $w$-freeness under the uniform distribution, which is proved in~\cite[Sec. 3.1]{RR-toct}.\footnote{Indeed, Ron and Rosin note that: ``The characterization may be useful for proving further results regarding property testing of
subsequence-freeness, as well as (sublinear) distance approximation.''}
In order to state their characterization, we introduce the notion of copies of $w$ in $T$, and more specifically, role-disjoint copies.

A \textit{copy} of $w = w_1\dots w_k$ in $T = t_1\dots t_n$ is a sequence of indices $(j_1,\dots,j_k)$ such that $1\leq j_1 < \dots < j_k \leq n$ and
$t_{j_1}\dots t_{j_k} = w$. It will be convenient to represent a copy as an array $C$ of size $k$ where $C[i] = j_i$.
A set of copies $\{C_\ell\}$ is said to be \textit{role-disjoint} if for every $i \in [k]$, the indices in $\{C_\ell[i]\}$ are distinct (though it is possible that $C_\ell[i] = C_{\ell'}[i']$ for $i\neq i'$ (and $\ell \neq \ell'$)). In the special case where the symbols of $w$ are all different from each other, a  set of copies is role disjoint simply if it consists of disjoint copies. Ron and Rosin prove~\cite[Theorem 3.4 $+$ Claim 3.1]{RR-toct} that $\Dist(T,w)$ equals the maximum number of role-disjoint copies of $w$ in $T$, divided by $n$.

Note that the analysis of the sample complexity of
one-sided error sample-based testing of subsequence-freeness
translates to bounding the size of the sample that is sufficient and necessary for ensuring that the sample contains evidence that $T$ is not $w$-free when $\Delta(T,w)>\epsilon$. Here evidence is in the form of a copy of $w$ in the sample, so that the testing algorithm simply checks whether such a copy exists.
On the other hand, the question of distance-approximation has a more algorithmic flavor, as it is not determined by the problem what must be done by the algorithm given a sample.

Focusing first on the uniform case, Ron and Rosin used their characterization (more precisely, the direction by which if $\Dist(T,w) > \eps$, then $T$ contains more than $\eps n$ role-disjoint copies of $w$), to prove that a sample of size $\Theta(k/\eps)$ contains at least one copy of $w$ with probability at least $2/3$. 
In this work we go further by designing an algorithm that actually approximates the number of role-disjoint copies of $w$ in $T$ (and hence approximates $\Dist(T,w)$), given a uniformly selected sample from $T$. 
It is worth noting that the probability of obtaining a copy in the sample might be quite different for texts that have \emph{exactly the same} number of role-disjoint copies of $w$ (and hence the same distance to being $w$-free).\footnote{For example, consider $w=1\dots k$, $T_1 = (1\dots k)^{n/k}$ and  $T_2 = 1^{n/k}\dots k^{n/k}$.} 

In the next subsection we discuss the aforementioned algorithm (for the uniform case), and in the following one address the distribution-free case. 

\vspace{-1.5ex}
\subsubsection{The uniform case}\label{subsub:intro-unif}
Let $\rdr(T,w)$ denote the number of role-disjoint copies of $w$ in $T$.
In a nutshell, the algorithm works by computing estimates of the numbers of occurrences of  symbols of $w$ in a relatively small number of prefixes of $T$, and using them to derive an estimate of $\rdr(T,w)$.
The more precise description of the algorithm and its analysis are based on several combinatorial claims that we present and which we discuss shortly next.

Let $\rdr_i^j(T,w)$ denote the number of role-disjoint copies of the length-$i$ prefix of $w$, $w_1\dots w_i$, in the length-$j$ prefix of $T$, $t_1\dots t_j$, and let $\nmr_i^j(T,w)$ denote the number of occurrences of the symbol $w_i$ in $t_1\dots t_j$.
In our first combinatorial claim, we show that for every $i\in [k]$ and $j\in [n]$,  the value of $\rdr_i^j(T,w)$ can be expressed  in terms of
the values of $\nmr_i^{j'}(T,w)$ for $j' \in [j]$ (in particular, $\nmr_i^j(T,w)$) and the values of $\rdr_{i-1}^{j'-1}(T,w)$ for $j' \in [j]$.
In other words, we establish a recursive expression which implies that if we know what are $\rdr_{i-1}^{j'-1}(T,w)$ and $\nmr_i^{j'}(T,w)$ for every $j' \in [j]$, then we can compute $\rdr_i^j(T,w)$ (and as an end result, compute $\rdr(T,w) = \rdr_k^n(T,w)$).

In our second combinatorial claim we show that if we only want an approximation of $\rdr(T,w)$, then it suffices to define (also in a recursive manner) a measure that depends on the values of
$\nmr_i^{j}(T,w)$ for every $i\in [k]$ but only for a relatively small number of choices of $j$, which are evenly spaced. To be precise, each such $j$ belongs to the set
$J =  \left\{r\cdot \gamma n\right\}_{r=1}^{1/\gamma} $ for $\gamma = \Theta(\delta/k)$.
We prove that since each interval $[(r-1)\gamma n+1,r\gamma n]$ is of size $\gamma n$ for this choice of $\gamma$, we can ensure that the aforementioned measure (which uses only $j\in J$) approximates $\rdr(T,w)$ to within $O(\delta n)$.

We then prove that if we replace each
$\nmr_i^{j}(T,w)$ for these choices of $j$ (and for every $i\in [k]$) by a sufficiently good estimate, then we incur a bounded error in the approximation of $\rdr(T,w)$. Finally, such estimates are obtained using (uniform) sampling, with a sample of size $\tilde{O}(k^2/\delta^2)$.

\vspace{-1.5ex}
\subsubsection{The distribution-free case}
In~\cite[Sec. 4]{RR-toct} it is shown that, given a word $w$, a text $T$ and a distribution $p$, it is possible to define a word $\tilw$ and a text $\tilT$ for which the following holds. 
First, $\Dist(T,w,p)$ is closely related to $\Dist(\tilT,\tilw)$. Second, the probability of observing a copy of $w$ in a sample selected from $T$ according to $p$ is closely related 
to the probability of observing a copy of $\tilw$ in a sample selected uniformly from $\tilT$.
%

We use the first relation stated above (i.e., between $\Dist(T,w,p)$ and $\Dist(\tilT,\tilw)$). However,  since we are interested in distance-approximation rather than one-sided error testing, the second relation stated above (between the
probability of observing a copy of  $w$ in $T$ and that of observing a copy of $\tilw$ in $\tilT$) is not sufficient for our needs, and we need to take a different (once again, more algorithmic) path, as we explain shortly next.

Ideally, we would have liked to sample uniformly from $\tilT$, and then run the algorithm discussed in the previous subsection using this sample (and $\tilw$). However, we only have sampling access to $T$ according to the underlying  distribution $p$, and we do not have direct sampling access to uniform samples from $\tilT$. Furthermore, since $\tilT$ is defined based on (the unknown) $p$, it is not clear how to determine the aforementioned subset of (evenly spaced) indices $J$.


For the sake of clarity, we continue the current exposition while making two assumptions. The first is that the distribution $p$ is such that there exists a value $\beta$, such that $p_j/\beta$ is an integer for every $j\in [n]$ (the value of $\beta$ need not be known).  The second is that in $w$ there are no two consecutive symbols that are the same. Under these assumptions, $\tilT = t_1^{p_1/\beta} \dots t_n^{p_n/\beta}$, $\tilw = w$, and $\Dist(\tilT,\tilw) = \Dist(T,w,p)$ (where $t_j^x$ for an integer $x$ is the subsequence that consists of $x$ repetitions of $t_j$).

Our algorithm for the distribution-free case (working under the aforementioned assumptions), starts by taking a sample distributed according to $p$ and using it to select a (relatively small) subset of indices in $[n]$.
Denoting these indices by $b_0,b_1,\dots,b_\ell$, where $b_0=0 < b_1 < \dots < b_{\ell-1} < b_\ell=n$, we would have liked to ensure that the weight according to $p$ of each interval $[b_{u-1}+1,b_u]$ is approximately the same (as is the case when considering the intervals defined by the subset $J$ in the uniform case). To be precise, we would have liked each interval to have relatively small weight, while the total number of intervals is not too large. However, since it is possible that for some single indices $j\in [n]$, the probability
$p_j$ is large, we also allow  intervals with large weight, where these intervals consist of a single index (and there are few of them). 

The algorithm next takes an additional sample, to approximate, for each $i\in [k]$ and $u\in [\ell]$, the weight, according to $p$, of the occurrences of the symbol $w_i$ in the length-$b_u$ prefix of $T$. Observe that prefixes of $T$ correspond to prefixes of $\tilT$. Furthermore, the weight according to $p$ of occurrences of symbols in such prefixes, translates to numbers of occurrences of symbols in the corresponding prefixes in $\tilT$, normalized by the length of $\tilT$.
The algorithm then uses these approximations
to obtain an estimate of $\Dist(\tilT,\tilw)$.

We note that some pairs of consecutive prefixes in $\tilT$ might be far apart, as opposed to what we had in the algorithm for the uniform case described in Section~\ref{subsub:intro-unif}. However, this is always due to single-index intervals in $T$ (for $j$ such that $p_j$ is large). Each such interval corresponds to a consecutive subsequence in $\tilT$ with repetitions of the same symbol, and we show that no additional error is incurred because of such intervals.

\vspace{-1ex}
\subsection{Related results}\label{subsec:intro-related}
As we have previously mentioned, the work most closely related to ours is that of Ron and Rosin on distribution-free sample-based testing of subsequence-freeness~\cite{RR-toct}. For other related results on property testing (e.g., testing other properties of sequences, sample-based testing of other types of properties and distribution-free testing (possibly with queries)), see the introduction of~\cite{RR-toct}, and in particular Section~1.4. 
For another line of work, on  sublinear approximation of
the longest increasing subsequence, see~\cite{NV-LIS} and references within.
Here we shortly discuss related results on distance approximation / tolerant testing.

As already noted, distance approximation and tolerant testing were first formally defined in~\cite{PRR}, and were shown to be significantly harder for some properties in~\cite{FF,BFLR}. Almost all previous results are query-based, and where the distance measure is with respect to the uniform distribution. These include~\cite{GuRu,FN,ACCL,MR,FR,CGR13,hoppen2017estimating,blais2019tolerant,levi2018lower,fiat2021efficient,pallavoor2022approximating}.
Kopparty and Saraf~\cite{kopparty2009tolerant} present results for query-based tolerant testing of linearity under several families of distributions.
Berman, Raskhodnikova and Yaroslavtsev~\cite{berman2014lp} give tolerant (query based) $L_p$-testing algorithms for monotonicity.
Berman, Murzbulatov and Raskhodnikova~\cite{berman2016tolerant} give a sample-based distance-approximation algorithms for image properties that works under the uniform distribution.


Canonne et al.~\cite{canonne2019testing} study the property of $k$-monotonicity of Boolean functions over various posets. A Boolean function over a finite poset domain $D$ is $k$-monotone if it alternates between the values $0$ and $1$ at most $k$ times on any ascending chain in $D$. For the special case of $D = [n]$, the property of $k$-monotonicity is equivalent to being free of $w$ of length $k+2$ where $w_1 \in \{0,1\}$ and $w_i = 1 - w_{i-1}$ for every $i \in [2,k+2]$.
One of their results implies an upper bound of 
$\widetilde{O}\left(\frac{k}{\delta^3}\right)$
on the sample complexity 
of
distance-approximation for $k$-monotonicity of functions $f:[n] \rightarrow \{0,1\}$ under the uniform distribution (and hence for $w$-freeness when $w$ is a binary subsequence of a specific form). 
This result generalizes to $k$-monotonicity in higher dimensions (at an exponential cost in the dimension $d$).

Blum and Hu~\cite{BH} study distance-approximation for $k$-interval (Boolean) functions over the line in the distribution-free active  setting. In this setting, an algorithm gets an unlabeled sample
\ifnum\icalpsub=0
 from the domain of the function,
 \fi
 and asks queries on a subset of sample points. Focusing on the sample complexity, they show that for any underlying distribution $p$ on the line,
a sample of size $\widetilde{O}\left(\frac{k}{\delta^2}\right)$ is sufficient for
approximating  the distance to being a $k$-interval function up to an additive error of $\delta$.
This implies a sample-based distribution-free distance-approximation algorithm with the same sample complexity for the special case of being free of the same pair of $w$'s described in the previous paragraph, replacing $k+2$ by $k+1$.

Blais, Ferreira Pinto Jr. and Harms~\cite{BFH} introduce a variant of the VC-dimension and use it to prove lower and upper
bounds on the sample complexity of distribution-free testing 
for a variety of properties. In particular,
one of their results implies that 
the linear dependence on  $k$ in the result of~\cite{BH} is essentially optimal.

Finally we mention that our procedure in the distribution-free case for constructing ``almost-equal-weight'' intervals by sampling is somewhat reminiscent of techniques used in other contexts of testing  when dealing with non-uniform distributions~\cite{black2020domain,harms2022downsampling,braverman2022improved}.





\vspace{-1ex}
\subsection{Further research}
The main open problem left by this work is  closing the gap between the upper and lower bounds that we give, and in particular understanding the precise dependence on $k$, or possibly other parameters determined by $w$ (such as $\kd$).
One step in this direction can be found in the Master Thesis of the first author~\cite{CS23}.

\vspace{-1ex}
\subsection{Organization}
In Section~\ref{sec:uni} we present our algorithm for distance-approximation under the uniform distribution.
\ifnum\icalpsub=0
The algorithm for the distribution-free case appears in Section~\ref{sec:dist-free-pref}.
In Section~\ref{sec:lower-bound} we prove our lower bound. In the appendix we provide Chernoff bounds and a few proofs of technical claims.
\else
Some of the main details of the  distribution-free case appears in Section~\ref{sec:dist-free-pref}, and 
in Section~\ref{sec:lower-bound} we prove our lower bound.
All missing details and proofs can be found in the full version of this paper~\cite{CSR23}.
\fi

%% file: uniform-pref.tex
\section{Distance approximation under the uniform distribution} \label{sec:uni}

In this section, we  address the problem of distance approximation when the underlying distribution is the uniform distribution.
As mentioned in the introduction, Ron and Rosin showed~\cite[Thm. 3.4]{RR-toct}
that $\Dist(T,w)$
(the distance of $T$ from $w$-freeness under the uniform distribution),
equals the number of role-disjoint copies of $w$ in $T$, divided by $n = |T|$ (where role-disjoint copies are as defined in the introduction -- see Section~\ref{subsec:ideas}). We may use $T[j]$ to denote the $j^{\th}$ symbol of $T$ (so that $T[j]=t_j$).

\smallskip
We start by introducing the following notations.

\begin{definition}\label{def:nm-and-rd}
For every $i\in [k]$ and $j\in [n]$,
let $\nmr_i^{j}(T,w)$ denote the number of occurrences of the symbol $w_i$ in the length $j$ prefix of $T$, $T[1,j]=T[1]\dots T[j]$.\footnote{Indeed, if $w_i = w_{i'}$ for $i\neq i'$, then $\nmr_i^{j}(T,w) = \nmr_{i'}^{j}(T,w)$ for every $j$.} Let $\rdr_{i}^{j}(T,w)$ 
denote the number of role-disjoint copies of the subsequence $w_1\dots w_i$ in $T[1,j]$.
When $i=k$ and $j=n$, we use the shorthand $\rdr(T,w)$ for $\rdr_k^n(T,w)$ (the total number of role-disjoint copies of $w$ in $T$).
\end{definition}
Observe that $\rdr_{1}^{j}(T,w)$ equals $\nmr_{1}^{j}(T,w)$ for every $j\in [n]$.

\medskip
Since, as noted above, $\Dist(T,w) = \rdr(T,w)/n$,
we would like to estimate $\rdr(T,w)$. More precisely, given $\delta>0$ we would like to obtain an estimate $\widehat{\rdr}$, such that: $\left|\widehat{\rdr} - \rdr(T,w) \right| \leq \delta n$.
To this end, we first establish two combinatorial claims. The first claim shows that the value of each $\rdr_i^j(T,w)$ can be expressed in terms of the values of $\nmr_i^{j'}(T,w)$ for $j' \in [j]$ (in particular, $\nmr_i^j(T,w)$) and the values of $\rdr_{i-1}^{j'-1}(T,w)$ for $j' \in [j]$. In other words, if we know what are $\rdr_{i-1}^{j'-1}(T,w)$ and $\nmr_i^{j'}(T,w)$ for every $j' \in [j]$, then we can compute $\rdr_i^j(T,w)$.
\begin{claim} \label{clm:n[i][j]}
For every $i \in \{2,\dots,k\}$ and $j\in [n]$,
\[
\rdr_{i}^{j}(T,w) = \nmr_i^{j}(T,w) - \max_{ j' \in [j]}\left\{ \nmr_i^{j'}(T,w) - \rdr_{i-1}^{j'-1}(T,w) \right\}\;.
\]
\end{claim}

\sloppy
Clearly, $\rdr_{i}^{j}(T,w) \leq \nmr_i^{j}(T,w)$ (for every $i \in \{2,\dots,k\}$ and $j\in [n]$), since each role-disjoint copy of $w_1\dots w_i$ in $T[1,j]$ must end with a distinct occurrence of $w_i$ in $T[1,j]$.
Claim~\ref{clm:n[i][j]} states by exactly how much is $\rdr_{i}^{j}(T,w)$ smaller than $\nmr_i^{j}(T,w)$.
Roughly speaking, the expression $\max_{ j' \in [j]}\left\{ \nmr_i^{j'}(T,w) - \rdr_{i-1}^{j'-1}(T,w) \right\}$ accounts for the number of occurrences of $w_i$ in $T[1,j]$ that cannot be used in role-disjoint copies of $w_1\dots w_i$ in $T[1,j]$.

\smallskip
\begin{proof}
For simplicity (in terms of notation), we prove the claim for the case that $i=k$ and $j=n$. The proof for general $i\in \{2,\dots,k\}$ and $j\in [n]$ is essentially the same up to renaming of indices.
Since $T$ and $w$ are fixed throughout the proof, we shall use the shorthand $\nmr_i^j$ for $\nmr_i^j(T,w)$ and
$\rdr_i^j$ for $\rdr_i^j(T,w)$.

For the sake of the analysis,
we start by describing a simple greedy procedure, that
constructs  $\rdr= \rdr_{k}^{n}$ role-disjoint copies of $w$ in $T$. 
The correctness of this procedure follows from~\cite[Claim 3.5]{RR-toct} and a simple inductive argument 
\ifnum\icalpsub=1
(details are provided in the full version of the paper~\cite{CSR23}).
\else
(for details see Appendix~\ref{app:missing}).
\fi
Every copy $C_m$, for  $m\in [\rdr]$ is an array of size $k$ whose values are monotonically increasing, where for every $i\in [k]$ we have that
 $C_m[i] \in [n]$, and $T[C_m[i]] = w_{i}$. Furthermore, for every $i\in [k]$  the indices $C_1[i],\dots,C_{\rdr}[i]$ are distinct. 
 For every $m = 1,\dots,\rdr$ and $i = 1,\dots,k$, the procedure scans $T$, starting from $T[C_m[i-1]+1]$ (where we define $C_m[0]$ to be $0$) and ending at $T[n]$ until it finds the first index $j$ such that $T[j]=w_{i}$ and $j \notin \{C_1[i],\dots,C_{m-1}[i]\}$. It then sets $C_m[i]=j$. For $i>1$ we say in such a case that the procedure \emph{matches} $j$ to the partial copy $C_m[1],\dots,C_m[i-1]$.

%
%

For $i \in [k]$, define: $G_i = \{j\in [n] \,:\, T[j]=w_i\}$. Also define: $G_i^+ = \{j \in G_i \,:\,\exists m, C_m[i] = j\}$ and $G_i^- = \{j \in G_i \,:\,\nexists m, C_m[i]=j\}$ (recall that $C_m[i]$
is the $i$-th index in the $m$-th greedy copy).

It is easy to verify that $\left|G_i\right| = \nmr_{i}^{n}$, $\left|G_i^+\right| = \rdr_{i}^{n}$ and $\left|G_i\right| = \left|G_i^+\right| + \left|G_i^-\right|$. To complete the proof, we will show that $\left|G_i^-\right| = \max_{j\in [n]}\left\{ \nmr_i^{j} - \rdr_{i-1}^{j-1} \right\}$.

Let $j^*$ be an index $j$ that maximizes  $\left\{ \nmr_i^{j} - \rdr_{i-1}^{j-1} \right\}$. In the interval $[j^*]$ we have  $\nmr_i^{j^*}$ occurrences of $w_i$, and in the interval $[j^{*}-1]$ we only have $\rdr_{i-1}^{j^*-1}$ role-disjoint copies of $w_1\dots w_{i-1}$. This implies that in the interval $[j^{*}]$ there are at least $\nmr_i^{j^*} - \rdr_{i-1}^{j^*-1}$ occurrences of $w_i$ that cannot be the $i$-th index of any greedy copy, and so we have
\begin{equation}
 \left|G_i^-\right| \ge \nmr_i^{j^*} - \rdr_{i-1}^{j^*-1} = \max_{j\in [n]}\left\{ \nmr_i^{j} - \rdr_{i-1}^{j-1} \right\}\;.
\end{equation}

On the other hand, denote by $j^{**}$ the largest index in $G_i^-$. 
Since each index $j \in [j^{**}]$ such that $T[j]=w_i$ is either the $i$-th element of some copy or is not the $i$-th element of any copy, $\nmr_i^{j^{**}} = \rdr_{i}^{j^{**}-1} + \left|G_i^-\right|$. We claim that $\rdr_{i}^{j^{**}-1} = \rdr_{i-1}^{j^{**}-1}$. Otherwise, $\rdr_{i}^{j^{**}-1} < \rdr_{i-1}^{j^{**}-1}$, in which case the index $j^{**}$ would have to be the the $i$-th element of a greedy copy. Hence, 
\begin{equation}
 \left|G_i^-\right| = \nmr_i^{j^{**}} - \rdr_{i-1}^{j^{**}-1} \leq \max_{j\in [n]}\left\{ \nmr_i^{j} -  \rdr_{i-1}^{j-1} \right\}\;.
 \end{equation}
 In conclusion, 
 \begin{equation}
 \left|G_i^-\right| = \max_{j\in [n]}\left\{ \nmr_i^{j} - \rdr_{i-1}^{j-1} \right\}\;,
 \end{equation}
 and the claim follows.
\end{proof}

In order to state our next combinatorial claim, we first introduce one more definition, which will play a central role  in obtaining an estimate for $\rdr(T,w)$. 
\begin{definition} \label{def:nJ}
For $\ell \leq n$, let $\nms$ be a $k\times \ell$ matrix of non-negative numbers, where we shall use $\nms_i^r$ to denote $\nms[i][r]$.
For every $r\in [\ell]$ let $\rd_1^r(\nms) = \nms_1^r$, and for every $i \in \{2,\dots,k\}$, let
\[
\rd_i^r(\nms) \eqdef \nms_i^r - \max_{r'\leq r} \left\{\nms_i^{r'} - \rd_{i-1}^{r'}(\nms)\right\}\;.
\]
When $i=k$ and $r=\ell$ we use the shorthand $\rd(\nms)$ for $\rd_k^\ell(\nms)$.
\end{definition}

In our second combinatorial claim we show that
for an appropriate choice of a matrix $\nms$, whose entries are a subset of all values
in $\left\{\nmr_i^j(T,w)\right\}_{i\in [k]}^{j\in [n]}$,
we can bound the difference between
$\rd(\nms)$  and $\rdr(T,w)$.
We later use sampling to obtain an estimated version of $\nms$.
\begin{claim} \label{clm:nJ-n}
Let $J = \left\{j_0,j_1, \dots, j_{\ell}\right\}$ be a set of indices satisfying $j_0=0< j_1 < j_2 <\dots < j_\ell=n$. 
Let
$\nms = \nms(J,T,w)$ be the matrix whose entries are $\nms_i^r = \nmr_i^{j_r}(T,w)$,
for every $i \in [k]$ and $r \in [\ell]$. Then we have 
\[
    \left|\rd(\nms) - \rdr(T,w)\right| \leq (k-1)\cdot \max_{\tau \in [\ell]}\left\{j_{\tau} - j_{\tau-1}\right\}\;.
    \]
\end{claim}
\begin{proof}
Recall that $\rd(\nms) = \rd_{k}^{\ell}(\nms)$ and $\rdr(T,w) = \rdr_{k}^{j_{\ell}}(T,w)$.
We shall prove that for every $i \in [k]$ and for every $r\in [\ell]$,
$\left|\rd_{i}^{r}(\nms) - \rdr_{i}^{j_{r}}(T,w)\right| \leq (i-1)\cdot \max_{\tau \in [r]}\left\{j_{\tau} - j_{\tau-1}\right\}$.
We prove this by induction on $i$.


\noindent
For $i=1$ and every $r\in [\ell]$,
\begin{eqnarray}
\left|\rd_{1}^r(\nms) - \rdr_{1}^{j_r}(T,w)\right| = \left| \nmr_1^{j_r}(T,w) -  \nmr_1^{j_r}(T,w) \right|
 =  0  \leq (1-1) \cdot \max_{\tau \in [1]}\left\{j_{\tau} - j_{\tau-1}\right\} \;,
\end{eqnarray}
where the first equality follows from the setting of $\nms$ and the definitions of $\rd_1^r(\nms)$ and $\rdr_1^{j_r}(T,w)$.

\noindent
For the induction step, we assume the claim holds for $i-1 \geq 1$ (and every $r\in [\ell]$)
and prove it for $i$. We have,
\begin{eqnarray}
\lefteqn{    \rd_{i}^{r}(\nms) - \rdr_{i}^{j_r}(T,w) } \nonumber \\
 &=& \nmr_i^{j_r}(T,w) - \max_{b \in [r]}\left\{ \nmr_i^{j_b}(T,w) - \rd_{i-1}^{b}(\nms) \right\} 
    - \rdr_{i}^{j_r}(T,w) \label{eq:j-to-J-1} \\
  &=& \max_{j \in [j_r]}\left\{ \nmr_i^{j}(T,w) - \rdr_{i-1}^{j-1}(T,w) \right\}
     - \max_{b\in [r]}\left\{ \nmr_i^{j_b}(T,w) - \rd_{i-1}^{b}(\nms) \right\} \label{eq:j-to-J-2}\;,
\end{eqnarray}
where Equation~\eqref{eq:j-to-J-1} follows from the setting of $\nms$ and the definition of $\rd_{i}^{r}(\nms)$,
and Equation~\eqref{eq:j-to-J-2} is implied by Claim~\ref{clm:n[i][j]}. Denote by
$j^*$ an index $j \in [j_r]$ that maximizes the first max term and
let $b^*$ be
the largest index such that
$j_{b^*} \leq j^*$.
We have:
\begin{eqnarray}
    \lefteqn{\max_{ j\in[j_r]}\left\{ \nmr_i^{j}(T,w) - \rdr_{i-1}^{j-1}(T,w) \right\} - \max_{b \in [r]}\left\{ \nmr_i^{j_b}(T,w) - \rd_{i-1}^{b}(\nms) \right\}} \nonumber \\
    &\leq&  \nmr_i^{j^*}(T,w) - \rdr_{i-1}^{j^*-1}(T,w) - \nmr_i^{j_{b^*}}(T,w) + \rd_{i-1}^{b^*}(\nms) \nonumber\\
    & =&  \nmr_i^{j^*}(T,w) + \rdr_{i-1}^{j_{b^*}}(T,w) - \rdr_{i-1}^{j_{b^*}}(T,w) - \rdr_{i-1}^{j^*-1}(T,w) - \nmr_i^{j_{b^*}}(T.w) 
      + \rd_{i-1}^{b^*}(\nms) \nonumber \\
    &\leq& \left(\rd_{i-1}^{b^*}(\nms) -  \rdr_{i-1}^{j_{b^*}}(T,w)\right) +  \left(\nmr_i^{j^*}(T,w) - \nmr_i^{j_{b^*}}(T,w) \right) 
    +\left(\rdr_{i-1}^{j_{b^*}}(T,w) - \rdr_{i-1}^{j^*-1}(T,w)\right) 
    \nonumber\\
    &\leq& (i-2) \max_{\tau \in [r]}\left\{j_{\tau} - j_{\tau-1}\right\} + \left(j^{*} - j^{b^{*}}\right) +\left(j^{b^{*}} - (j^{*}-1)\right)
    \label{eq:induction}\\
    &=& (i-2) \max_{\tau \in [r]}\left\{j_{\tau} - j_{\tau-1}\right\} + 1
    \nonumber\\
    &\leq& (i-2) \max_{\tau \in [r]}\left\{j_{\tau} - j_{\tau-1}\right\} + \max_{\tau \in [r]}\left\{j_{\tau} - j_{\tau-1}\right\} \nonumber \\
    &=& (i-1)\max_{\tau \in [r]}\left\{j_{\tau} - j_{\tau-1}\right\}\label{eq:max-j-max-b}\;,
\end{eqnarray}
where in Equation~\eqref{eq:induction}
we used the induction hypothesis.
By combining Equations~\eqref{eq:j-to-J-2} and~\eqref{eq:max-j-max-b}, we get that
\begin{equation}
\rd_{i}^{r}(\nms) - \rdr_{i}^{j_r}(T,w) \leq (i-1)\max_{\tau \in [r]}\left\{j_{\tau} - j_{\tau-1}\right\}\;. \label{eq:nms-T-w}
\end{equation}

Similarly to Equation~\eqref{eq:j-to-J-2},
\begin{equation}
    \rdr_{i}^{j_r}(T,w) - \rd_{i}^{r}(\nms)
     = \max_{b\in [r]}\left\{ \nmr_i^{j_b}(T,w) - \rd_{i-1}^{b}(\nms) \right\} 
     - \max_{j\in [j_r]}\left\{ \nmr_i^{j}(T,w) - \rdr_{i-1}^{j-1}(T,w) \right\}\;.
     \label{eq:J-2-to-j}
\end{equation}

Let $b^{**}$ be 
the index $b\in [r]$ that maximizes the first max term.
We have
\begin{eqnarray}
    \lefteqn{\max_{b \in [r]}\left\{ \nmr_i^{j_b}(T,w) - \rd_{i-1}^{b}(\nms) \right\} - \max_{ j \in [j_r]}\left\{ \nmr_i^{j}(T,w) - \rdr_{i-1}^{j-1}(T,w) \right\}} \nonumber \\
    & \leq& \nmr_i^{j_{b^{**}}}(T,w) - \rd_{i-1}^{{b^{**}}}(\nms) - \nmr_i^{j_{b^{**}}}(T,w) + \rdr_{i-1}^{j_{b^{**}}-1}(T,w) \nonumber\\
    & \leq&  \rdr_{i-1}^{j_{b^{**}}}(T,w) - \rd_{i-1}^{b^{**}}(\nms) 
    \;\leq\;  \left|\rdr_{i-1}^{j_{b^{**}}}(T,w) - \rd_{i-1}^{b^{**}}(\nms)\right| \nonumber\\
    & \leq& (i-2) \max_{\tau \in [r]}\left\{j_{\tau} - j_{\tau-1}\right\}
    \;\leq\; (i-1) \max_{\tau \in [r]}\left\{j_{\tau} - j_{\tau-1}\right\} \;.
     \label{eq:max-b-max-j}
\end{eqnarray}
Hence (combining Equations~\eqref{eq:J-2-to-j} and~\eqref{eq:max-b-max-j}),\footnote{It actually holds that $\rd_{i}^{r}(\nms) \geq \rdr_{i}^{j_r}(T,w)$, so that $\rdr_{i}^{j_r}(T,w) - \rd_{i}^{r}(\nms)\leq 0$, but for the sake of simplicity of the inductive argument, we prove the same upper bound on $\rdr_{i}^{j_r}(T,w) - \rd_{i}^{r}(\nms)$ as on $\rd_{i}^{r}(\nms) - \rdr_{i}^{j_r}(T,w)$.}
\begin{equation}
     \rdr_{i}^{j_r}(T,w) - \rd_{i}^{r}(\nms) \leq (i-1) \max_{\tau \in [r]}\left\{j_{\tau} - j_{\tau-1}\right\} \label{eq:T-w-nms} \;.
\end{equation}
Together, Equations~\eqref{eq:nms-T-w} and~\eqref{eq:T-w-nms} give us that
\begin{equation}
    \left|\rd_{i}^{r}(\nms) - \rdr_{i}^{j_r}(T,w)\right| \leq (i-1) \max_{\tau \in [r]}\left\{j_{\tau} - j_{\tau-1}\right\}\;,
\end{equation}
and the proof is completed.
\end{proof}

In our next claim we  bound the difference between $\rd(\widehat{\nms)} - \rd(\widetilde{\nms})$ for any two matrices (with dimensions $k\times \ell$), given a bound on the $L_\infty$ distance between them. We later apply this claim with $\widetilde{\nms} = \nms$ for $\nms$ as defined in Claim~\ref{clm:nJ-n}, and $\widehat{\nms}$ being a matrix that contains estimates
$\widehat{\nm}_{i}^{r}$ of $\nmr_{i}^{j_r}(T,w)$  (respectively).
We discuss how to obtain $\widehat{\nms}$ in Claim~\ref{clm:sample}.

\begin{claim} \label{clm:estimator}
Let $\gamma\in (0,1)$, and let $\widehat{\nms}$ and $\widetilde{\nms}$ be two $k\times \ell$ matrices.
If for every $i \in [t]$ and  $r \in [\ell]$,
\ifnum\icalpsub=0
\begin{equation}\nonumber
    \left|\widehat{\nms}_{i}^{r} - \widetilde{\nms}_{i}^{r} \right| \leq \gamma n\;,
\end{equation}
\else
$\left|\widehat{\nms}_{i}^{r} - \widetilde{\nms}_{i}^{r} \right| \leq \gamma n$,
\fi
then
\ifnum\icalpsub=0
\begin{equation}\nonumber
   \left|\rd(\widehat{\nms}) - \rd(\widetilde{\nms})\right| \leq (2k-1)\gamma n\;.
\end{equation}
\else
$\left|\rd(\widehat{\nms}) - \rd(\widetilde{\nms})\right| \leq (2k-1)\gamma n$.
\fi
\end{claim}
\begin{proof}
We shall prove that for every $t \in [k]$ and for every $r\in [\ell]$,
$\left|\rd_t^r(\widehat{\nms)} - \rd_t^r(\widetilde{\nms})\right| \leq (2t-1)\gamma n$.
We prove this by induction on $t$.\\
For $t = 1$ and every $r\in [\ell]$, we have
\begin{equation}
\left| \rd_1^r(\widehat{\nms}) - \rd_1^r(\widetilde{\nms}) \right| \;=\;
\left|\widehat{\nms}_{1}^{r} - 
                    \widetilde{\nms}_1^r \right|
     \;\leq\; \gamma n \;.
\end{equation}
Now assume the claim is true for $t-1 \geq 1$ and for every $r \in [\ell]$, and we prove it for $t$.
For any $r \in [\ell]$, by the definition of $\rd_t^r(\cdot)$,
\begin{eqnarray}
    \lefteqn{\left|\rd_{t}^{r}(\widehat{\nms}) - \rd_{t}^{r}(\widetilde{\nms})\right| } \nonumber \\
   &=&\left| \widehat{\nms}_t^{r} - \max_{r'' \in [r]}\left\{ \widehat{\nms}_t^{r''} - \rd_{t-1}^{r''}(\widehat{\nms})\right\} -
    \widetilde{\nms}_t^{r} + \max_{r' \in [r]}\left\{ \widetilde{\nms}_t^{r'} - \rd_{t-1}^{r'}(\widetilde{\nms})\right\}
    \right| \nonumber \\
    &\leq& \gamma n +
    \left| \max_{r' \in [r]}\left\{ \widetilde{\nms}_t^{r'} - \rd_{t-1}^{r'}(\widetilde{\nms})\right\} -
    \max_{r'' \in [r]}\left\{ \widehat{\nms}_t^{r''} - \rd_{t-1}^{r''}(\widehat{\nms})\right\}
    \right| \;,\label{eq:hat-tilde1}
    \end{eqnarray}
where in the last inequality we used the premise of the claim.
Assume that the first max term  in Equation~\eqref{eq:hat-tilde1} is at least as large as the second (the case that the second term is larger than the first is dealt with analogously), and
    let $r^*$ be the index that maximizes the first max term.
    Then,
    \begin{eqnarray}
        \lefteqn{\left| \max_{r' \in [r]}\left\{ \widetilde{\nms}_t^{r'} - \rd_{t-1}^{r'}(\widetilde{\nms})\right\} - \max_{r'' \in [r]}\left\{ \widehat{\nms}_t^{r''} - \rd_{t-1}^{r''}(\widehat{\nms})\right\}
        \right|} \nonumber \\
        &\leq&  \left|\left(\widetilde{\nms}_t^{r^*} - \widehat{\nms}_t^{r^*}\right) + \left(\rd_{t-1}^{r^*}(\widehat{\nms}) - \rd_{t-1}^{r^*}(\widetilde{\nms})\right) \right| \nonumber\\
        &\leq &  \left|\widetilde{\nms}_t^{r^*} - \widehat{\nms}_t^{r^*}\right| + \left|\rd_{t-1}^{r^*}(\widehat{\nms}) - \rd_{t-1}^{r^*}(\widetilde{\nms})\right| \nonumber\\
    &\leq &  \gamma n + (2t-3) \gamma n = (2t-2) \gamma n \;,
    \label{eq:hat-tilde2}
\end{eqnarray}
where we used 
the premise of the claim once again, and the induction hypothesis.
The claim follows by combining Equation~\eqref{eq:hat-tilde1} with
Equation~\eqref{eq:hat-tilde2}.
\end{proof}

The next 
claim states that we can obtain  good estimates for all values in $\left\{\nmr_{i}^{j_r}(T,w)\right\}_{i\in [k]}^{r\in [\ell]}$
(with a sufficiently large sample).
Its (standard) proof is deferred to 
\ifnum\icalpsub=0 
Appendix~\ref{app:missing}.
\else
the full version of this paper~\cite{CSR23}.
\fi
\begin{claim}\label{clm:sample}
For any $\gamma \in (0,1)$ and $J = \{j_1,\dots,j_\ell\}$ (such that $1 \leq j_1 < \dots < j_\ell=n$),
by taking a sample of size $\Theta\left(\frac{\log(k\cdot \ell)}{\gamma^2}\cdot \right)$ from $T$, we can obtain with probability at least $2/3$ estimates
$\left\{\widehat{\nms}_i^r\right\}_{i\in [k]}^{r\in [\ell]}$, such that
\begin{equation} \label{eq:hatn_n}
    \left|\widehat{\nms}_{i}^{r} - \nmr_{i}^{j_r}(T,w) \right| \leq \gamma n\;,
\end{equation}
for every $i\in[k]$ and $r\in [\ell]$.
\end{claim}

We can now restate and prove our main theorem for distance approximation under the uniform distribution.

\ifnum\icalpsub=1
\smallskip\noindent
 \textcolor{lipicsGray}{$\blacktriangleright$}
\nobreakspace
{\sffamily\bfseries{Theorem~\ref{thm:uni_alg}}}.
\else
\medskip\noindent
\textbf{Theorem~\ref{thm:uni_alg}}~
\fi
\emph{\UniAlgThmStatement}

\medskip\noindent
While our focus is on the sample complexity of the algorithm, we note that its running time is linear in the size of the sample.

\smallskip
\begin{proof}
The algorithm sets $\gamma = \delta/(3k)$ and $J = \left\{\gamma n, 2\gamma n, \dots, n\right\}$.
It first applies Claim~\ref{clm:sample} with the above setting of $\gamma$ to obtain the estimates $\left\{\widehat{\nms}_i^r\right\}$ for every $i\in [k]$ and $r\in [\ell]$,
which with probability at least $2/3$ are as stated in Equation~\eqref{eq:hatn_n}.
If we take $\widetilde{\nms} = \nms$ for $\nms$ as defined in Claim~\ref{clm:nJ-n}, then the premise of
Claim~\ref{clm:estimator} holds. We can hence apply Claim~\ref{clm:estimator}, and combining with Claim~\ref{clm:nJ-n} and the definition of $J$, we get that with probability at least $2/3$, for the  matrix $\widehat{\nms}$,
\begin{equation}
\left|\rd(\widehat{\nms}) - \rdr(T,w) \right|
\leq (2k-1)\gamma n + (k-1) \gamma n = (3k-2) \gamma n \leq \delta n\;.
\end{equation}
The algorithm hence computes 
$\rd(\widehat{\nms}) = \rd_k^\ell(\widehat{\nms})$
in an iterative manner, based on Definition~\ref{def:nJ},
and outputs $\wDist = \rd(\widehat{\nms})/n$.
Since $\rdr(T,w)/n = \Dist(T,w)$, the theorem follows.
\end{proof}

%% file: dist-free-pref.tex
 \section{Distribution-free distance approximation}
 \label{sec:dist-free-pref}

As noted in the introduction, our algorithm for approximating the distance from subsequence-freeness under a general distribution $p$ works by reducing the problem to approximating the distance from subsequence-freeness under the uniform distribution.
However, we won't be able to use the algorithm presented in Section~\ref{sec:uni} as is. There are two main obstacles, explained shortly next.
In the reduction, given a word $w$ and access to samples from a text $T$, distributed according to $p$, we define a word $\widetilde{w}$ 
and a text $\widetilde{T}$ such that if we can obtain a good approximation of $\Dist(\widetilde{T},\widetilde{w})$ then we get a good approximation of $\Dist(T,w,p)$. (Recall that $\Dist(T,w,p)$ denotes the distance of $T$ from being $w$-free under the distribution $p$.)
 However, first, we don't actually have direct access to uniformly distributed samples from $\widetilde{T}$, and second, we
 cannot work with a set $J$ of indices that induce equally sized intervals (of a bounded size), as we did in  Section~\ref{sec:uni}.

 We address these challenges (as well as precisely define $\tilT$ and $\tilw$) in several stages. We start, in Sections~\ref{subsec:intervals} and~\ref{subsec:est-symb-dens}, by using sampling according to $p$, in order to construct intervals in $T$ that have certain properties (with sufficiently high probability).
 The role of these intervals will become clear
 \ifnum\icalpsub=0
 in the following other subsections.
 \else
 as we proceed.
 Due to space constraints, several proofs are deferred to the 
 full version of this paper~\cite{CSR23}.
 \fi

\subsection{Interval construction and classification}\label{subsec:intervals}
We begin this subsection by defining intervals in $[n]$ that are determined by $p$ (which is unknown to the algorithm).
We then construct intervals by sampling from $p$, where the latter intervals are in a sense approximations of the former (this will be formalized subsequently). Each constructed interval will be classified as either ``heavy'' or ``light'', depending on its (approximated) weight according to $p$.  Ideally, we would have liked all intervals to be light, but not too light, so that their number won't be too large (as was the case when we worked under the uniform distribution and simply defined intervals of equal size). However, for a general distribution $p$ we might have single indices $j \in [n]$ for which $p_j$ is large, and hence we also need to allow heavy intervals (each consisting of a single index).
We shall make use of the following two definitions.

\begin{definition}\label{def:wt-p}
    For any two integers $j_1 \leq j_2$, let
    $[j_1,j_2]$ denote the interval $\{j_1,\dots,j_2\}$. For every $j_1, j_2 \in [n]$, define
\ifnum\icalpsub=1
$\wt_p([j_1,j_2]) \eqdef \sum_{j=j_1}^{j_2}p_j$
\else
    \begin{equation} \nonumber 
        \wt_p([j_1,j_2]) \eqdef \sum_{j=j_1}^{j_2}p_j
    \end{equation}
\fi
    to be the weight of the interval $[j_1,j_2]$ according to $p$.
    We shall use the shorthand $\wt_p(j)$ for $\wt_p([j,j])$.
\end{definition}

\begin{definition}\label{def:wt-S}
    Let $S$ be a multiset of size $s$, with elements from $[n]$. For every $j \in [n]$, let $\nm_S(j)$ be the number of elements in $S$ that equal $j$. For every $j_1, j_2 \in [n]$, define
    \ifnum\icalpsub=1
      $\wt_{S}([j_1,j_2]) \eqdef \frac{1}{s}\sum_{j=j_1}^{j_2}\nm_S(j)$
    \else
    \begin{align} \nonumber 
        \wt_{S}([j_1,j_2]) \eqdef \frac{1}{s}\sum_{j=j_1}^{j_2}\nm_S(j)
    \end{align}
    \fi
    to be the estimated weight of the interval $[j_1,j_2]$ according to  $S$.
    We shall use the shorthand $\wt_S(j)$ for $\wt_S([j,j])$.
\end{definition}

In the next definition, and the remainder of this section, we shall use
\begin{equation} \label{eq:z_eq}
z \;= \; c_{z}\frac{k}{\delta}\;,
\end{equation}
where let $c_z = 100$.

We next define the aforementioned set of intervals, based on $p$.
Roughly speaking, we try to make the intervals as equally weighted as possible, keeping in mind that some indices might have a large weight, so we  assign each to an interval of its own.
\begin{definition} \label{def:real_interval_construction}
Define a sequence of indices in the following iterative manner.
Let $h_0=0$ and for $\ell = 1,2,\dots$, as long as $h_{\ell-1} < n$, let $h_\ell$ be defined as follows.
If $\wt_{p}(h_{\ell-1}+1) > \frac{1}{8z}$, then $h_\ell = h_{\ell-1}+1$. Otherwise, let $h_\ell$ be the maximum index $h'_\ell \in [h_{\ell-1}+1,n]$ such that $\wt_{p}([h_{\ell-1}+1,h_{\ell}']) \leq \frac{1}{4z}$ and for every $h''_\ell \in [h_{\ell-1}+1,h'_\ell]$, $\wt_{p}(h''_\ell) \leq \frac{1}{8z}$. Let $L$ be such that $h_L = n$.

Based on the indices $\{h_\ell\}_{\ell=0}^L$ defined above, for every $\ell \in [L]$, let $H_\ell = [h_{\ell-1}+1,h_\ell]$ and let $\mathcal{H} = \left\{H_\ell\right\}_{\ell=1}^{L}$. We partition $\mathcal{H}$ into three subsets as follows.
Let $\calHsin$ be the subset of all $H \in \mathcal{H}$ such that $|H| = 1$ and $\wt_p(H) > \frac{1}{8z}$. Let $\calHmed$ be the set of all $H \in \mathcal{H}$ such that $|H| \neq 1$ and $\frac{1}{8z} \leq \wt_p(H) \leq \frac{1}{4z}$. Let $\calHsml$ be the set of all $H \in \mathcal{H}$ such that $\wt_p(H) < \frac{1}{8z}$.
\end{definition}
Observe that since $\wt_p(T) = 1$, then $|\calHsin \cup \calHmed| \leq 8z$. In addition, since between each $H', H'' \in \calHsml$ there has to be at least one  $H \in \calHsin$, then we also have $|\calHsml| \leq 8z + 1$.


\smallskip
By its definition, $\mathcal{H}$ is determined by $p$. We next  construct a set of intervals $\mathcal{B}$ based on sampling according to $p$ (in a similar, but not identical, fashion to Definition~\ref{def:real_interval_construction}).
Consider a sample $S_1$ of size $s_1$
selected according to $p$ (with repetitions), where $s_1$ will be set subsequently.

%
\begin{definition} \label{def:interval_construction}
Given a sample $S_1$ (multiset of elements in $[n]$) of size $s_1$, determine a sequence of indices in the following iterative manner.
Let $b_0=0$ and for $u = 1,2,\dots$, as long as $b_{u-1} < n$, let $b_u$ be defined as follows.
If $\wt_{S_1}(b_{u-1}+1) > 1/z$, then $b_u = b_{u-1}+1$. Otherwise, let $b_u$ be the maximum index $b'_u \in [b_{u-1}+1,n]$ such that $\wt_{S_1}([b_{u-1}+1,b_{u}']) \leq \frac{1}{z}$. Let $U$ be such that $b_U = n$.

Based on the indices $\{b_u\}_{u=0}^U$ defined above, for every $u \in [U]$, let $B_u = [b_{u-1}+1,b_u]$, and let $\mathcal{B} = \left\{B_u\right\}_{u=1}^{U}$. For every $u \in [U]$, if $\wt_{S_1}(B_{u}) > \frac{1}{z}$, then we say that $B_u$ is \textsf{heavy}, otherwise it is \textsf{light}.
\end{definition}
Observe that each heavy interval consists of a single element.

In order to relate between $\calH$ and $\mathcal{B}$, we introduce the
 following event, based on the sample $S_1$.

 \begin{definition} \label{def:E1}
      Denote by $E_1$ the event where
    \begin{align}
        &\forall H \in \calHsin \cup \calHmed,\;\;  \frac{1}{2}\wt_{p}(H) \leq \wt_{S_1}(H) \leq \frac{3}{2}\wt_{p}(H) \label{eq:E1-sin-med}\;,\\
        &\forall H \in \calHsml,\;\;  \wt_{S_1}(H) \leq \frac{1}{2z}\;.
        \label{eq:E1-sml}
    \end{align}
 \end{definition}

\begin{claim} \label{clm:E_1}
If the size of the sample $S_1$ is $s_1 = 120 z \log(240z)$, then
\ifnum\icalpsub=1
$\Pr\left[E_1\right] \ge \frac{8}{10}$,
\else
    \begin{equation}\nonumber
        \Pr\left[E_1\right] \ge \frac{8}{10}\;,
    \end{equation}
\fi
 where the probability is over the choice of $S_1$.
\end{claim}

\ifnum\icalpsub=0
\begin{proof}
Recall that $\wt_p(H) \geq \frac{1}{8z}$ for every $H \in \calHsin \cup \calHmed$.
    Using the multiplicative Chernoff bound (see Theorem \ref{thm:Chernoff}) we get that for every $H \in \calHsin \cup \calHmed$
    \begin{equation}
        \Pr\left[\left|\wt_{S_1}(H) - \wt_{p}(H)\right| > \frac{1}{2}\wt_{p}(H)\right] < 2\exp\left(-\frac{1}{12} \wt_{p}(H) s_1\right) < \frac{1}{120z}\;.
    \end{equation}

 Next recall that $\wt_p(H) \leq \frac{1}{4z}$ for every $H \in \calHsml$.
    Define the random variables $\left\{\chi_r\right\}_{r=1}^{s_1}$ such that for every $r \in [s_1]$, $\chi_r = 1$ with probability $\frac{1}{4z}$ and $0$ otherwise.
    Once again using the multiplicative Chernoff bound we get that for every $H \in \calHsml$
    \begin{eqnarray}
        \Pr\left[\wt_{S_1}(H) > 2\frac{1}{4z}\right] &\leq& \Pr\left[\sum_{r=1}^{s_1} \chi_r > 2\frac{1}{4z}\right] < \exp\left(-\frac{1}{3} \frac{1}{4z} s_1\right) \nonumber \\ &=& \exp\left(-\frac{1}{12z} s_1\right) < \frac{1}{240z}\;.
    \end{eqnarray}
    Using a union bound over all $H \in \calHsin  \cup \calHmed \cup \calHsml$ 
    (recall that by the discussion following Definition~\ref{def:real_interval_construction},
    $|\calHsin \cup \calHmed| \leq 8z$ and $|\calHsml| \leq 8z + 1$),
    we get
    \begin{equation}
        \Pr[E_1] \ge 1 - 8z \cdot \frac{1}{120z} - 9z \cdot \frac{1}{240z} \ge \frac{8}{10}\;,
    \end{equation}
 and the claim is established.
\end{proof}
\fi

\begin{claim} \label{clm:init_est}
    Conditioned on the event $E_1$,  for every  $u \in [U]$ such that
$B_u$ is light, 
$\wt_{p}(B_u) < \frac{6}{z}$.
\end{claim}

\ifnum\icalpsub=0
\begin{proof}
Consider an interval $B_u$ that is light.
Let $\mathcal{H}_4$ be the minimal set such that $\mathcal{H}_4 = \mathcal{H}_1 \cup \mathcal{H}_2 \cup \mathcal{H}_3$ where $\mathcal{H}_1 \subseteq \calHsin$, $\mathcal{H}_2 \subseteq \calHmed$, $\mathcal{H}_3 \subseteq \calHsml$ and $B_u \subseteq \mathcal{H}_4$. Let $\mathcal{H}_5 $ be the smallest subset of $\mathcal{H}_{4}$ such that $\mathcal{H}_4 \setminus \mathcal{H}_5 \subseteq B_u $. It is easy to verify that $\mathcal{H}_5 \subseteq \mathcal{H}_2 \cup \mathcal{H}_3$ and that $|\mathcal{H}_5| \leq 2$.

Conditioned on $E_1$ (Definition~\ref{def:E1},  Equation~\eqref{eq:E1-sin-med}), we have that  $\wt_p(H) < 2\wt_{S_1}(H)$ for every $H \in \mathcal{H}_1 \cup \mathcal{H}_2$, and hence also
$\wt_{S_1}(H) > \frac{1}{16z}$ for every $H \in \calH_1$. Since $B_u$ is light,  $\wt_{S_1}(B_u) < \frac{1}{z}$, implying that $|\mathcal{H}_1| \leq 17$. As mentioned before, there has to be at least one interval $H \in \calHsin$ between any pair of intervals $H', H'' \in \calHsml$, implying that $|\mathcal{H}_3| \leq |\mathcal{H}_1| + 2 \leq 19$. Therefore,
\begin{align}
    \wt_p(B_u) &\leq \sum_{H \in \mathcal{H}_4} \wt_p(H) = \sum_{H \in \mathcal{H}_1} \wt_p(H) + \sum_{H \in \mathcal{H}_2} \wt_p(H) + \sum_{\mathcal{H} \in H_3} \wt_p(H) \\
    &\leq  2\sum_{H \in \mathcal{H}_1} \wt_{S_1}(H) + 2\sum_{H \in \mathcal{H}_2} \wt_{S_1}(H) + |\mathcal{H}_3| \frac{1}{8z}
    \leq 2\wt_{S_1}(B_u) + 2\sum_{H \in \mathcal{H}_5} \wt_{S_1}(H) +  \frac{19}{8z} \\
    &\leq  2\wt_{S_1}(B_u) + 4 \max_{H \in \mathcal{H}_2 \cup \mathcal{H}_3}\left\{\wt_{S_1}(H)\right\} + \frac{19}{8z} \\
    &\leq \frac{2}{z} + \frac{1}{z} + \frac{19}{8z} < \frac{6}{z}\;,
\end{align}
and the claim follows.
\end{proof}
\fi

\subsection{Estimation of symbol density and weight of intervals}\label{subsec:est-symb-dens}
In this subsection we estimate the weight, according to $p$, of every interval $[b_u]$ for $u \in U$, as well as its symbol density, focusing on symbols that occur in $w$. Note that $[b_u]$ is the union of the intervals $B_1,\dots,B_u$. We first introduce some notations.

    For any word $w^{*}$, text $T^{*}$, $i \in [|w^{*}|]$ and $j\in [|T^{*}|]$, let $I_{i}^{j}(T^{*},w^{*}) = 1$ if $T^{*}[j] = w^{*}_i$ and $0$ otherwise.
  We next set 
    \begin{equation}\label{eq:xi-i-u}
        \xi_{i}^{u} =
            \sum_{j \in [b_u]} I_i^j(T,w) p_j\;.
    \end{equation}

Consider a sample $S_2$ of size $s_2$ selected according to $p$ (with repetitions), where $s_2$ will be set subsequently.
For every $u \in [U]$ and $i \in [k]$, set 
    \begin{equation}\label{eq:xi_est_def}
        \breve{\xi}_{i}^{u} =
          \frac{1}{s_2} \sum_{j \in [b_u]}I_i^j(T,w) \nm_{S_2}(j)\;.
    \end{equation}

\begin{definition} \label{def:E2}
The     event $E_2$ (based on $S_2$) is defined as follows. For every $i \in [k]$ and $u \in [U]$,
    \begin{equation}\label{eq:E2-hat-xi-xi}
             \left|\breve{\xi}_{i}^{u} - \xi_{i}^{u} \right| \leq \frac{1}{z} \;,
    \end{equation}
    and for every $u \in [U]$
    \begin{equation}\label{eq:wt-S2-wt-p}
            \left|\wt_{S_2}([b_u]) - \wt_p([b_u]) \right| \leq \frac{1}{z}\;.
    \end{equation}
\end{definition}

\begin{claim} \label{clm:E_2}
If the size of the sample $S_2$ is $s_2 = z^2 \log\left(40k U \right)$, then
\ifnum\icalpsub=1
$ \Pr\left[E_2\right] \ge \frac{9}{10}$,
\else
    \begin{equation}\nonumber
        \Pr\left[E_2\right] \ge \frac{9}{10}\;,
    \end{equation}
\fi
    where the probability is over the choice of $S_2$.
\end{claim}

\ifnum\icalpsub=0
\begin{proof}
\sloppy
    Using the additive Chernoff bound (see Theorem~\ref{thm:Chernoff}) along with the fact that $\E\left[\frac{\nm_{S_2}(j)}{s_2}I_i^j(T,w)\right] = p_j I_i^j(T,w)$ yields the following.
    \begin{align}
        \Pr\left[\left|\breve{\xi}_{i}^{u} - \xi_{i}^{u} \right| > \frac{1}{z} \right] &= \Pr\left[\left|\frac{1}{s_2} \sum_{j \in [b_u] }I_i^j(T,w) \nm_{S_2}(j) - \sum_{j \in [b_u] }I_i^j(T,w) p_j \right| > \frac{1}{z} \right] \\
        &< 2 \exp(-2 \frac{1}{z^2} s_2) \leq \frac{1}{20 k U}\;.
    \end{align}
    By applying a union bound over all $i \in [k]$ and $u \in [U]$, we get that with probability of at least $\frac{19}{20}$, $\left|\breve{\xi}_{i}^{u} - \xi_{i}^{u} \right| \leq \frac{1}{z}$.
    Another use of the additive Chernoff bound along with the fact that $\E\left[\frac{\nm_{S_2}(j)}{s_2}\right] = p_j$ gives us that
    \begin{align}
        &\Pr\left[\left|\wt_{S_2}([b_u]) - \wt_p([b_u]) \right| > \frac{1}{z} \right] = \Pr\left[\left|\frac{1}{s_2} \sum_{j \in [b_u] } \nm_{S_2}(j) - \sum_{j \in [b_u] } p_j \right| > \frac{1}{z} \right] \\
        &\leq 2 \exp(-2 \frac{1}{z^2} s_2) \leq \frac{1}{20U}\;.
    \end{align}

    \sloppy
    Again using a union bound over all $u \in [U]$, we get that with probability of at least $\frac{19}{20}$ we have $\left|\wt_{S_2}([b_u]) - \wt_p([b_u]) \right| \leq \frac{1}{z}$. One last use of the union bound gives us that $\Pr\left[E_2\right] \ge \frac{9}{10}$
\end{proof}
\fi

\subsection{Reducing from distribution-free to uniform}

In this subsection we give
\ifnum\icalpsub=0
the
details
\fi
 for  aforementioned reduction from the distribution-free case to the uniform case, using the intervals and estimators that were defined in the previous subsections.
We start by providing three definitions, taken from~\cite{RR-toct}, which will be used  in the reduction. The first two definitions are for the notion of \textit{splitting} (variants of this notion were also used in previous works, e.g.,~\cite{DK16}).

\begin{definition} \label{def:split_T}
    For a text $T = t_1 \dots t_n$, a text $\widetilde{T}$ is said to be a splitting of $T$ if $\widetilde{T} = t_1^{\alpha_1} \dots t_n^{\alpha_n}$ for some $\alpha_1 \dots \alpha_n \in \mathbb{N}^{+}$. We denote by $\phi$ the splitting map, which maps each (index of a) symbol of $\widetilde{T}$ to its origin in $T$. Formally, $\phi : [|\widetilde{T}|] \rightarrow [n]$ is defined as follows. For every $\ell \in [|\widetilde{T}|] = [\sum_{i=1}^{n} \alpha_i]$, let $\phi(\ell)$ be the unique $i \in [n]$ that satisfies $\sum_{r=1}^{i-1}\alpha_r < \ell < \sum_{r=1}^{i}\alpha_r$.
\end{definition}
%
Note that by this definition, $\phi$ is a non-decreasing surjective map, satisfying $\widetilde{T}[\ell] = T[\phi(\ell)]$ for every $\ell \in [|\widetilde{T}|]$. For a set $S \subseteq [|\widetilde{T}|]$ we let $\phi(S) = \left\{\phi(\ell) : \ell \in S \right\}$. With a slight abuse of notation, for any $i \in [n]$ we use $\phi^{-1}(i)$ to denote the set $\left\{\ell \in [|\widetilde{T}|] : \phi(\ell) = i\right\}$, and for a set $S \subseteq [n]$ we let $\phi^{-1}(S) = \left\{\ell \in [|\widetilde{T}|] : \phi(\ell) \in S \right\}$
\begin{definition} \label{def:split_p}
    Given  text $T = t_1 \dots t_n$ and a corresponding probability distribution $p = (p_1, \dots, p_n)$, a splitting of $(T,p)$ is a text $\widetilde{T}$ along with a corresponding probability distribution $\hat{p} = (\hat{p}_1, \dots, \hat{p}_{|\widetilde{T}|})$, such that $\widetilde{T}$ is a splitting of $T$ and $\sum_{\ell \in \phi^{-1}(i)} \hat{p}_{\ell} = p_i$ for every $i \in [n]$.
\end{definition}

The third definition is of a set of words, where no two consecutive symbols are the same.
\begin{definition}\label{def:Wc}
    Let $\calW_c = \left\{ w \,:\, w_{j+1} \neq w_j, \forall j \in [k-1] \right\}$\;.
\end{definition}

 \subsubsection{A basis for reducing from distribution-free to uniform}
\label{sec:dist-app-xi}
Let $\tilw$ be a word of length $\tilk$ and $\tilT$ a text of length $\tiln$.
In this subsection we establish a claim, which gives sufficient conditions on a (normalized version) of an estimation matrix $\hatnms$, under which it can be used to obtain an estimate of $\Dist(\tilT,\tilw)$ with a small additive error.

We first state a claim that is similar to Claim~\ref{clm:nJ-n}, with a small, but important difference, that takes into account intervals in $\widetilde{T}$ (determined by a set of indices $J$) that consist of repetitions of a single symbol.
\ifnum\icalpsub=0
Since its proof is very similar to the proof of Claim~\ref{clm:nJ-n},  it is deferred to Appendix~\ref{app:missing}.
\fi
Recall that $\rd(\cdot)$ was defined in
Definition~\ref{def:nJ}, and that $R(\tilT,\tilw)$ denotes the number of role-disjoint copies of $\tilw$ in $\tilT$.

\begin{claim} \label{clm:nJ-n-J-prime}
Let $J = \left\{j_0,j_1,  \dots, j_{\ell}\right\}$ be a set of indices satisfying $j_0=0 < j_1 < j_2 <\dots < j_\ell=\widetilde{n}$. 
Let $\nms$ be the matrix whose entries are $\nms_i^r = \nmr_i^{j_r}(\widetilde{T},\widetilde{w})$ for every $i \in [\widetilde{k}]$ and $r \in [\ell]$.
Let $J' = \{r\in [\ell]\,:\,\widetilde{T}[j_{r-1}+1]=\dots=\widetilde{T}[j_r]\}$. Then
\[
    \left|\rd(\nms) - \rdr(\widetilde{T},\widetilde{w})\right| \leq (\widetilde{k}-1)\cdot \max_{r \in [\ell]\setminus J'}\left\{(j_{r} - j_{r-1})\right\}\;.
    \]
\end{claim}

The following observation can be easily proved by induction.
\begin{observation} \label{obs:linearity_of_M}
    Let $\widehat{\nms}$ be a matrix of size $\widetilde{k} \times \ell$. Then
    \begin{equation}
        \frac{1}{\widetilde{n}}\rd(\widehat{\nms}) = \rd\left(\frac{\widehat{\nms}}{\widetilde{n}}\right)\;.
    \end{equation}
\end{observation}

The next claim
will serve as the basis for our reduction from the general, distribution-free case, to the uniform case.

\begin{claim} \label{clm:alg_no_n}
Let $\widehat{\nms}$ be a $\widetilde{k} \times \ell$ matrix, $J = \left\{j_0,j_1, j_2, \dots, j_{\ell}\right\}$ be a set of indices satisfying $j_0=0 < j_1 < j_2 <\dots < j_\ell=\widetilde{n}$ and 
let $c_1$ and $c_2$ be constants.
Suppose that the following conditions are satisfied.
\begin{enumerate}
    \item \label{itm:J_not_spaced} For every $r \in [\ell]$, if $j_r - j_{r - 1} > c_1 \cdot\frac{\delta \widetilde{n}}{\widetilde{k}}$,
    then $\widetilde{T}[j_{r-1}+1]=\dots=\widetilde{T}[j_{r}]$.
    \item \label{itm:good_approx} For every $i \in [\widetilde{k}]$ and  $r\in [\ell]$, $\left|\widehat{\nms}_{i}^{r} - \nmr_{i}^{j_r}(\widetilde{T},\widetilde{w}) \right| \leq c_2 \cdot \frac{\delta \widetilde{n}}{\widetilde{k}}$.
\end{enumerate}
Then,
\begin{equation}\nonumber
    \left|\rd\left(\frac{\widehat{\nms}}{\widetilde{n}}\right) - \Dist(\widetilde{T},\widetilde{w})\right| \leq (c_1 + 2c_2) \delta\;.
\end{equation}\nonumber
\end{claim}

\ifnum\icalpsub=0
\begin{proof}
    Let $\nms$ be the matrix whose entries are $\nms_i^r = \nmr_i^{j_r}(\widetilde{T},\widetilde{w})$ for every $i \in [\widetilde{k}]$ and $r \in [\ell]$. We use Claim \ref{clm:nJ-n-J-prime} and Item \ref{itm:J_not_spaced} in the premise of the current claim to obtain that $\left|\rd(\nms) - \rdr(\widetilde{T},\widetilde{w})\right| \leq c_1 \delta \widetilde{n}$. We also use Claim \ref{clm:estimator} and Item \ref{itm:good_approx} in the premise of the current claim to obtain that $\left|\rd(\widehat{\nms}) - \rd(\nms)\right| \leq 2 c_2 \delta \widetilde{n}$. Combining these bounds we get that $\left|\rd(\widehat{\nms}) - \rdr(\widetilde{T},\widetilde{w})\right| \leq (c_1+2c_2) \delta \widetilde{n}$. The claim follows by applying Observation \ref{obs:linearity_of_M} along with the fact that $\frac{\rdr(\widetilde{T},\widetilde{w})}{\widetilde{n}} = \Dist(\widetilde{T},\widetilde{w})$.
\end{proof}
\fi

\subsubsection{Establishing the reduction for $w\in \calW_c$ and quantized $p$}
\ifnum\icalpsub=0
For ease of readability, we begin by addressing the special case in which $w \in \calW_c$ (recall Definition~\ref{def:Wc})
and where there exists $\beta \in (0,1)$ such that $p_j/\beta$ is an integer for every $j \in [n]$.
We later show how to deal with the general case, where we rely on techniques from~\cite{RR-toct} and introduce some new ones that are needed for implementing our algorithm.
\else
For the ease of readability, in this subsection we address the special case in which $w \in \calW_c$ (recall Definition~\ref{def:Wc}), and
in 
full version of this paper~\cite{CSR23}
we show
how to deal with the general case.
\fi

For the case considered in this subsection, let $\widetilde{T} = t_{1}^{\alpha_1}\dots t_{n}^{\alpha_{n}}$ where $\alpha_j = \frac{p_{j}}{\beta}$ for every $j \in [n]$, so that $|\widetilde{T}|=\frac{1}{\beta}$. Define $\widetilde{p}$ by $\widetilde{p}_j = \beta$ for every $j \in [|\widetilde{T}|]$, so that $\widetilde{p}$ is the uniform distribution. Since $p_j = \beta \cdot \alpha_j$, for every $j \in [n]$, we get that $(\widetilde{T},\widetilde{p})$ is a splitting of $(T,p)$
(recall Definition~\ref{def:split_p}), and hence
by \cite[Clm. 4.4]{RR-toct} (using the assumption that $w\in \calW_c$),
\begin{equation}\label{eq:nus}
    \Dist(\widetilde{T},w,\widetilde{p}) = \Dist(T,w,p)\;.
\end{equation}

Denote $\widetilde{n} = |\widetilde{T}|$. 
We begin by defining a set of intervals of $[\widetilde{n}]$, where $\{b_0,\dots,b_U\}$ and $\mathcal{B} = \{B_1,\dots,B_U\}$ are as defined in Section~\ref{subsec:intervals}, and $\phi$ is as in Definition~\ref{def:split_p}.
\begin{definition}\label{def:Btil}
Let $\widetilde{b}_0 = 0$, and for every $u \in [U]$, let $\widetilde{b}_u = \max \left\{h \in [\widetilde{n}] : \phi(h) = b_u\right\}$.
For every $u \in [U]$ let $\widetilde{B}_u = [\widetilde{b}_{u-1}+1,\widetilde{b}_u]$, and
define  $\mathcal{\widetilde{B}} = \left\{\widetilde{B}_u\right\}_{u=1}^{U}$\;.
\end{definition}

We next introduce a notation for the weights, according to $\widetilde{p}$, of unions of these intervals.
    For every $i \in [k]$ and $u \in [U]$,
    \begin{equation}
        {\widetilde{\xi}}_{i}^{u} = \sum_{j \in [\widetilde{b}_u]} I_{i}^{j}(\widetilde{T},w) \widetilde{p}_j\;.
    \end{equation}
    Note that
    \begin{equation} \label{eq:xi_tilde_wrap}
    \widetilde{\xi}_{i}^{u} = \frac{1}{\widetilde{n}}\nmr_{i}^{b_u}(\widetilde{T},w)\;.
    \end{equation}

\begin{claim} \label{clm:comma_approx_simple}
For every $i \in [k]$ and $u \in [U]$
\ifnum\icalpsub=1
$\widetilde{\xi}_{i}^{u} = \xi_{i}^{u}$,
\else
    \begin{equation}\nonumber
        \widetilde{\xi}_{i}^{u} = \xi_{i}^{u}\;,
    \end{equation}
\fi
where $\xi_i^u$ is as defined in Equation~\eqref{eq:xi-i-u}.
\end{claim}
\ifnum\icalpsub=0
\begin{proof}
    \begin{align}
        \xi_{i}^{u} &= \sum_{j \in [b_u]}I_{i}^{j}(T,w) p_j = \sum_{j \in [b_u]}I_{i}^{j}(T,w) \sum_{\widetilde{j} \in \phi^{-1}(j)}\widetilde{p}_{\widetilde{j}} \\
        &= \sum_{j \in [b_u]} \sum_{\widetilde{j} \in \phi^{-1}(j)} I_{i}^{j}(T,w) \widetilde{p}_{\widetilde{j}} = \sum_{j \in [b_u]} \sum_{\widetilde{j} \in \phi^{-1}(j)} I_{i}^{\widetilde{j}}(\widetilde{T},w) \widetilde{p}_{\widetilde{j}} = \sum_{\widetilde{j} \in \widetilde{b}_u} I_{i}^{\widetilde{j}}(\widetilde{T},w) \widetilde{p}_{\widetilde{j}} = {\widetilde{\xi}}_{i}^{u}\;,
    \end{align}
    and the claim is established.
\end{proof}
\fi

We can now state and prove the following lemma.
\begin{lemma} \label{lem:wrap-up-w-Wc}
    Let $w$ be a word of length $k$ in $\calW_c$, $T$ a text of length $n$, and $p$ a distribution over $[n]$ for which there exists $\beta \in (0,1)$ such that $p_j/\beta$ is an integer for every $j\in [n]$.
    There exists an algorithm that, given a parameter $\delta \in (0,1)$,
    takes a sample of size
$\Theta\left(\frac{k^2}{\delta^2}\cdot \log\left(\frac{k}{\delta}\right)\right)$ from $T$, distributed according to $p$, and outputs an estimate $\wDist$
such that $|\wDist - \Dist(T,w,p)| \leq \delta$ with probability at least $2/3$.
\end{lemma}
As in the uniform case, the running time of the algorithm is linear in the size of the sample. 

\smallskip
\begin{proof}
    The algorithm first takes a sample $S_1$ of size $s_1 = 120 z \log(240z)$ and constructs a set of intervals $\mathcal{B}$ as defined in Definition~\ref{def:interval_construction}.
    Next the algorithm takes another sample, $S_2$, of size $s_2 = z^2 \log(40 k U)$ according to which it defines an estimation matrix $\widehat{\xi}$ of size $k \times U$ as follows. For every $i \in [k]$ and $u \in [U]$, it sets $\widehat{\xi}[i][u] = \breve{\xi}_{i}^{u}$, where $\breve{\xi}_{i}^{u}$ is as defined in
    Equation~\eqref{eq:xi_est_def}. Lastly the algorithm outputs $\wDist = \rd(\widehat{\xi})$, where $\rd$ is as defined in Definition \ref{def:nJ}.

    We would like to apply Claim \ref{clm:alg_no_n} in order to show that $|\wDist - \Dist(\widetilde{T},w)| \leq \delta$ with probability of at least $\frac{2}{3}$.
    By the setting of $s_1$, applying Claim~\ref{clm:E_1} gives us that
    with probability at least $\frac{8}{10}$, the event $E_1$, as defined in Definition~\ref{def:E1}, holds.
    By the setting of $s_2$, applying Claim~\ref{clm:E_2} gives us that with probability at least $\frac{9}{10}$, the event $E_2$, as defined in Definition~\ref{def:E2},  holds. We henceforth condition on both events (where they  hold together with probability at least $7/10$).

    In order to apply Claim~\ref{clm:alg_no_n}, we set $\widetilde{w} = w$, $J = \left\{\widetilde{b}_0,\widetilde{b}_1,\dots,\widetilde{b}_U\right\}$ (recall Definition~\ref{def:Btil}) and $\widehat{\nms} = \widetilde{n} \widehat{\xi}$, for $\widehat{\xi}$ as defined above. Also, we set $c_1 = \frac{1}{2}$ and $c_2 = \frac{1}{4}$. We next show that both items in the premise of the claim are satisfied.

    To show that Item \ref{itm:J_not_spaced} is satisfied, we first note that since $\widetilde{p}$ is uniform, then for every $u \in U$, $\wt_{\widetilde{p}}(b_u) = \frac{\widetilde{b}_{u} - \widetilde{b}_{u-1}}{\widetilde{n}}$. We use the consequence of Claim \ref{clm:init_est} (recall that we condition on $E_1$) by which for every $u$ such that $\frac{\widetilde{b}_{u} - \widetilde{b}_{u-1}}{\widetilde{n}} \ge \frac{6}{z}$, $B_u$ is heavy (since for every $u \in U$, $\wt_{\widetilde{p}}(\widetilde{B}_u) = \wt_{p}(B_u)$). By Definition \ref{def:interval_construction} this implies that $B_u$ contains only one index, and so $\widetilde{T}[\widetilde{b}_{u-1}+1] = \dots = \widetilde{T}[\widetilde{b}_{u}]$. By the definition of $z$ (Equation~\eqref{eq:z_eq}) and the setting of $c_1$, the item is satisfied.

    To show that Item \ref{itm:good_approx} is satisfied, we use the definition of $E_2$ (Definition~\ref{def:E2},  Equation~\eqref{eq:E2-hat-xi-xi}) together with Claim \ref{clm:comma_approx_simple}, which give us $|\widehat{\xi}_{i}^{u} - \widetilde{\xi}_{i}^{u}| \leq \frac{1}{z}$ for every $i \in [k]$ and $u \in [U]$. By Equation \eqref{eq:xi_tilde_wrap}, the definition of $z$ and the setting of $c_2$, we get that the item is satisfied.

    After applying Claim~\ref{clm:alg_no_n} we get that $|\wDist - \Dist(\widetilde{T},w)| \leq ( c_1  + 2c_2)\delta$, which by the setting of $c_1$ and $c_2$ is at most $\delta$. Since $\widetilde{p}$ is the uniform distribution,  $\Dist(\widetilde{T},w) = \Dist(\widetilde{T},w,\widetilde{p})$ and since $\Dist(\widetilde{T},w,\widetilde{p}) = \Dist(T,w,p)$ (by Equation~\eqref{eq:nus}), the lemma follows.
\end{proof}

\ifnum\icalpsub=0
\medskip
In the next subsections we turn to the general case where we do not necessarily have that $w \in \calW_c$ or that there exists a value $\beta$ such that for every $j \in [n]$, $p_j/\beta$ is an integer. In this general case we need to take some extra steps until we can perform a splitting. Beginning with performing a reduction to a slightly different distribution, then performing a reduction to $w \in \calW_c$.
While this follows~\cite{RR-toct}, for the sake of our algorithm, along the way we need to show how to define the estimation matrix $\widehat{\xi}$ ($ = \hatnms/\tiln$) and the corresponding set of indices $J$ so that we can apply Claim~\ref{clm:alg_no_n}, similarly to what was shown in the proof of
Lemma~\ref{lem:wrap-up-w-Wc}.
\else
In 
full version of this paper~\cite{CSR23} we address the general case where we do not necessarily have that $w \in \calW_c$ or that there exists a value $\beta$ such that for every $j \in [n]$, $p_j/\beta$ is an integer.
\fi

\ifnum\icalpsub=0
\subsection{Quantized distribution}\label{subsec:quant-dist}
Let $\eta = c_{\eta}\frac{1}{n z}$, where $c_{\eta} = \frac{1}{16} $. We define $\Ddot{p}$ by ``rounding'' $p_j$ for every $j \in [n]$ to its nearest larger integer multiple of $\eta$. Namely, $\Ddot{p}_j = \left \lceil{\frac{p_j}{\eta}}\right \rceil \eta$, for every $j \in [n]$. By this definition, $L_1(\Ddot{p},p) = \sum_{j=1}^{n}|\Ddot{p}_j-p_j| \leq \eta n = \frac{c_{\eta}}{z}$. We  define $\Dot{p}$ to be a normalized version of $\ddot{p}$. That is, letting $\zeta = \frac{1}{\sum_{j=1}^{n}\Ddot{p}_j}$, we set $\Dot{p}_j = \zeta \ddot{p}_j$, for every $j \in [n]$, and note that $\zeta \leq 1$. Observe that $L_1(\Dot{p},\ddot{p}) = \sum_{j=1}^{n}|\zeta\ddot{p}_j-\ddot{p}_j| = \frac{|\zeta-1|}{\zeta} = \frac{1}{\zeta} - 1$.
Also observe that $\frac{1}{\zeta} = \sum_{j=1}^{n}(p_j+(\Ddot{p}_j-p_j)) = 1 + \sum_{j=1}^{n}(\ddot{p}_j - p_j)$. We have $L_1(\Dot{p},\ddot{p}) = \frac{1}{\zeta} - 1 = \sum_{j=1}^{n}(\ddot{p}_j - p_j) \leq \sum_{j=1}^{n}|\ddot{p}_j - p_j| = L_1(p,\ddot{p})$. Using the triangle inequality we get that
\begin{equation} \label{eq:L1-p-pdot}
   L_1(p,\Dot{p}) \leq L_1(p,\ddot{p}) + L_1(\ddot{p},\Dot{p}) \leq 2 L_1(p,\ddot{p}) \leq 2 \frac{c_{\eta}}{z}\;.
\end{equation}
By~\cite[Clm. 4.1]{RR-toct} we have that
\begin{equation} \label{eq:nu_prob}
        |\Dist(T,w,p) - \Dist(T,w,\Dot{p})| \leq L_1(p,\Dot{p})\;.
\end{equation}
Finally, note that for every $j \in [n]$, $\Dot{p}_j$ is an integer multiple of $\zeta \eta$, as $\Dot{p}_j = \zeta \eta \left \lceil{\frac{p_j}{\eta}}\right \rceil$.

    For every $i \in [k]$ and $u \in [U]$, set 
    \begin{equation}
        \dot{\xi}_{i}^{u} = \sum_{j \in [b_u] }I_{i}^{j}(T,w) \dot{p_j}\;.
    \end{equation}
\begin{claim} \label{clm:dot_approx}
For every $i \in [k]$ and $u \in [U]$,
    \begin{equation} \label{eq:xi_prob}
        \left|\dot{\xi}_{i}^{u} - \xi_{i}^{u}\right| \leq \sum_{j \in [b_u] }|\dot{p}_j-p_j|\;,
\end{equation}
and for every $u \in [U]$
    \begin{equation} \label{eq:omega_prob}
        |\wt_{\dot{p}}([b_u]) - \wt_{p}([b_u])| \leq \sum_{j \in [b_u]}|\dot{p}_j-p_j|\;.
    \end{equation}
\end{claim}

\begin{proof}
   Equation~\eqref{eq:xi_prob} follows by
   using the triangle inequality along with the fact the for every $j \in [n], i \in [k]$, $I_{i}^{j}(T,w) \leq 1$
    \begin{equation}
        \left|\dot{\xi}_{i}^{u} - \xi_{i}^{u}\right| = \left|\sum_{j \in [b_u] } I_{i}^{j}(T,w)(\dot{p}_j - p_j) \right| \leq \sum_{j \in [b_u] } \left|\dot{p}_j-p_j\right|\;.
    \end{equation}
Equation~\eqref{eq:omega_prob} follows
by the triangle inequality
    \begin{equation}
        |\wt_{\dot{p}}([b_u]) - \wt_{p}([b_u])| = \left|\sum_{j \in [b_u]}\dot{p}_j - p_j \right| \leq \sum_{j \in [b_u]}|\dot{p}_j-p_j| \;,
    \end{equation}
and the claim is established.
\end{proof}

\subsection{Dealing with $w \notin \calW_c$}
We would have liked to consider a $\dot{p}$-splitting of $T = t_1 t_2 \dots t_n$ and then use the relationship between the distance from $w$-freeness before and after the splitting. However, we only know this connection between the distances in the case of $w \in \calW_c$. Hence, we shall apply a reduction from a general $w$ to $w \in \calW_c$, as was done in~\cite{RR-toct}, in their proof of Lemma 4.8. Without loss of generality, assume $0$ is a symbol that does not appear in $w$ or $T$ (if there is no such symbol in $\Sigma$, then we extend $\Sigma$ to $\Sigma \cup \left\{0\right\}$). Let $w' = w_1 0 w_2 0\dots w_{k-1} 0 w_k 0$, $T' = t_1 0 t_2 0\dots t_n 0$ and $p' = (\dot{p}_1/2,\dot{p}_1/2,\dots ,\dot{p}_n/2,\dot{p}_n/2)$. Note that $w'$ is in $\calW_c$.
By~\cite[Clm. 4.6]{RR-toct},
\begin{equation} \label{eq:nu_half}
    \Dist(T',w',p') = \frac{1}{2} \Dist(T,w,\dot{p})\;.
\end{equation}

\smallskip\noindent
Here too we define a set of intervals of 
$[2n]$.
\begin{definition} \label{def:interval_construction_wc}
Let $b'_0 = 0$. Define $U'$, $\left\{b'_u\right\}_{u=1}^{U'}$ and the function $f: [U'] \rightarrow [U]$ using Algorithm~\ref{alg:b_prime}.
For every $u \in [U']$ let $B'_u = [b'_{u-1}+1,b'_u]$, and
define $\mathcal{B}' = \left\{B'_u\right\}_{u=1}^{U'}$\;.
\end{definition}

    \begin{algorithm}[H]
    \caption{} \label{alg:b_prime}
    \textbf{Input:} $U$, $\left\{b_v\right\}_{v=1}^{U}$, an indication for every $v \in U$ whether $B_v = \left\{b_v\right\}_{v=1}^{U}$ is heavy or light.\\
    \textbf{Output:} $U'$, $\left\{b_u'\right\}_{u=1}^{U'}$.
        \begin{algorithmic}[1]
            \STATE $u = 1$, $v = 1$
            \WHILE{$v \leq U$}
                \IF{$B_{v}$ is heavy}
                    \STATE $b'_{u} = 2b_{v} - 1$, $b'_{u+1} = 2b_{v}$
                    \STATE $f(u) = v$, $f(u+1) = v$
                    \STATE $v = v + 1$, $u = u + 2$
                    \ELSE
                    \STATE $b'_{u} = 2b_{v}$
                    \STATE $f(u) = v$
                    \STATE $v = v + 1$, $u = u + 1$
                \ENDIF
            \ENDWHILE
            \STATE $U' = \max \left\{f^{-1}(U)\right\}$
        \end{algorithmic}
    \end{algorithm}
    Intuitively, Algorithm \ref{alg:b_prime} makes sure that $0$'s that come after what is a heavy interval in $T$ become a single-index interval themselves in $T'$. On the other hand, the rest of the $0$'s are joined to the light interval that includes their left neighbour in $T$, to form a new interval, which will have the same weight as the light interval in $T$. It also sets the function $f$ that maps intervals in $T'$ to their corresponding intervals in $T$.

\medskip
\noindent
For every $i \in [2k]$ and $u \in [U']$, set
    \begin{equation}
        {\xi'_i}^{u} = \sum_{j \in [b'_u]} I_{i}^{j}(T',w') p'_j\;.
    \end{equation}

\begin{observation} \label{obs:tilde_approx}
    For every $u \in [U']$
    and for every $i \in [2k]$ such that $2 \not| \;i$ (meaning $w'_i \neq 0$),
    \begin{eqnarray}
        {\xi'_{i}}^{u} = \frac{1}{2}\dot{\xi}_{\frac{i+1}{2}}^{f(u)}\;,
    \end{eqnarray}
    whereas if $2 \mid i$ (meaning $w'_i = 0$),
    \begin{eqnarray}
        {\xi'_{i}}^{u} =
            \frac{1}{2} \begin{cases}
            \wt_{\dot{p}}\left([b_{f(u)}]\right) & \text{if } B_{f(u)} \text{is light} \\
            \wt_{\dot{p}}\left([b_{f(u)}]\right) & \text{if } B_{f(u)} \text{is heavy and } T'[b'_{u}] = 0 \\
            \wt_{\dot{p}}\left([b_{f(u)-1}]\right) & \text{if } B_{f(u)} \text{is heavy and } T'[b'_{u}] \neq 0\;.
            \end{cases}
    \end{eqnarray}
\end{observation}

\subsection{Uniform distribution via splitting}
Recall that $\eta$ and $\zeta$ were defined at the beginning of Section~\ref{subsec:quant-dist}.
Let $\widetilde{T} = {t'}_{1}^{\alpha_1}\dots {t'}_{2n}^{\alpha_{2n}}$ where $\alpha_j = \left\lceil{\frac{p_{j}}{\eta}} \right \rceil$ for every $j \in [2n]$. Define the distribution $\widetilde{p}$ by $\widetilde{p}_j = \frac{1}{2} \zeta \eta$ for every $j \in [|\widetilde{T}|]$, so that $\widetilde{p}$ is the uniform distribution (recall that $\zeta$ and $\eta$ where defined in Section \ref{subsec:quant-dist}). Since $p'_j = \frac{1}{2} \zeta \eta \cdot \alpha_j = \sum_{\tilde{j} \in \phi^{-1}(j)}\widetilde{p}_{\tilde{j}}$, for every $j \in [2n]$, we get that $(\widetilde{T},\widetilde{p})$ is a splitting of $(T',p')$ (recall Definition \ref{def:split_p}).

We make another use of \cite[Thm. 4.4]{RR-toct}, by which splitting preserves the distance from $w$-freeness, to establish that
\begin{equation} \label{eq:nu_split_wc}
    \Dist(\widetilde{T},w',\widetilde{p}) = \Dist(T',w',p')\;.
\end{equation}

Denote $\widetilde{n} = |\widetilde{T}| = \frac{2}{\zeta \eta}$.
We next define a set of intervals of $[\widetilde{n}]$.
\begin{definition} \label{def:Btil_wc}
Let $\widetilde{b}_0 = 0$, and for every $u \in [U']$, let \\ $\widetilde{b}_u = \max \left\{h \in [\widetilde{n}] : \phi(h) = b'_u\right\}$.
For every $u \in [U']$ let $\widetilde{B}_u = [\widetilde{b}_{u-1}+1,\widetilde{b}_u]$, and
define  $\mathcal{\widetilde{B}} = \left\{\widetilde{B}_u\right\}_{u=1}^{U'}$\;.
\end{definition}

\noindent
For for every $i \in [2k]$ and $u \in [U']$, set
    \begin{equation}
        \widetilde{\xi}_{i}^{u} = \sum_{j \in [\widetilde{b}_u]} I_{i}^{j}(\widetilde{T},w') \widetilde{p}_j\;,
    \end{equation}
and note that
    \begin{equation} \label{eq:xi_tilde_final}
        \widetilde{\xi}_{i}^{u} = \frac{1}{\widetilde{n}}\nmr_{i}^{\widetilde{b}_u}(\widetilde{T},w')\;.
    \end{equation}

The proof of the next claim is almost identical to the proof of Claim \ref{clm:comma_approx_simple}, and is hence omitted.
\begin{claim} \label{clm:comma_approx}
For every $i \in [2k]$ and $u \in [U']$,
    \begin{equation} \nonumber
        \widetilde{\xi}_{i}^{u} = {\xi'_{i}}^{u}\;.
    \end{equation}
\end{claim}

For the last claim in this subsection, recall that the event $E_1$ was defined in
Definition~\ref{def:E1}.
\begin{claim} \label{clm:light_b_is_light_b}
    Conditioned on the event $E_1$, for every $u \in [U']$, if $\mathcal{B}_{f(u)}$ is light, then $\wt_{\widetilde{p}}(\widetilde{\mathcal{B}}_{u}) < \frac{6}{z} + \frac{c_{\eta}}{z}$.
\end{claim}
\begin{proof}
Consider any $u \in [U']$ such that
    $\mathcal{B}_{f(u)}$ is light.
    Conditioned on the event $E_1$, which was defined in Definition \ref{def:E1}, the consequence of Claim \ref{clm:init_est} holds and so $\wt_p(\mathcal{B}_{f(u)}) < \frac{6}{z}$. Also, it is easy to verify that if $\mathcal{B}_{f(u)}$ is light then $\left|\wt_p(\mathcal{B}_{f(u)}) - \wt_{p'}(\mathcal{B}'_{u})\right| \leq L_1(p,\Dot{p})$. Since $\wt_{p'}(\mathcal{B}'_{u}) = \wt_{\widetilde{p}}(\widetilde{\mathcal{B}}_{u})$ and $L_1(p,\ddot{p}) \leq \frac{c_{\eta}}{z}$, the claim follows.
\end{proof}

\subsection{Estimators for the distribution-free case}
For every $i \in [2k]$ and $u \in [U']$, let $x(u,i)$ take the following values

\medskip\noindent
$x(u,i)=1$ if $2 \not| \;i$,\\
$x(u,i)=2$ if $2 \mid i$ and $B_{f(u)}$ is light,\\
$x(u,i)=2$ (also) if $2 \mid i$ and $B_{f(u)}$ is heavy and $T'[b'_{u}] = 0$,\\
$x(u,i)=3$ if $2 \mid i$ and $B_{f(u)}$ is heavy and $T'[b'_{u}] \neq 0$.\\

\noindent
    Define the following estimator.
    For every $i \in [2k]$ and $u \in [U']$
    \begin{equation} \label{eq:est_nm_prime}
        \widehat{\xi}_{i}^{u} =
            \frac{1}{2} \begin{cases}
            \breve{\xi}_{\frac{i+1}{2}}^{f(u)} & \text{if } x(u,i) = 1 \\
            \wt_{S_2}\left([b_{f(u)}]\right) & \text{if } x(u,i) = 2 \\
            \wt_{S_2}\left([b_{f(u-1)}]\right) & \text{if } x(u,i) = 3\;. \\
            \end{cases}
    \end{equation}

For the next claim, recall that the event $E_2$ was defined in Definition~\ref{def:E2}.
\begin{claim} \label{clm:approx_breve_xi}
    Conditioned on the event $E_2$, for every $i \in [2k]$ and  $u \in [U']$
    \begin{eqnarray} \label{eq:err_breve_xi}
        \left|\widehat{\xi}_{i}^{u} - \widetilde{\xi}_{i}^{u}\right| \leq \frac{c_{\eta}}{z} + \frac{1}{2z}\;.
    \end{eqnarray}
\end{claim}

\begin{proof}
    Using the triangle inequality, along with Claim~\ref{clm:dot_approx}, Observation~\ref{obs:tilde_approx} and Claim~\ref{clm:comma_approx}, we get that for every $i \in [2k]$ and $u \in [U']$
    \begin{eqnarray}
        \left|\widehat{\xi}_{i}^{u} - \widetilde{\xi}_{i}^{u}\right| &\leq& \frac{1}{2} \sum_{r \in B_{f(u)}}\left|\dot{p}_j - p_j\right| \\ \nonumber &+& \frac{1}{2}
        \begin{cases}
            \left|\breve{\xi}_{{\frac{i+1}{2}}}^{f(u)} - \xi_{\frac{i+1}{2}}^{f(u)}\right| & \text{if } x(u,i) = 1 \\
            \left|\wt_{S_2}\left([b_{f(u)}]\right) - \wt_p\left([b_{f(u)}]\right)\right| & \text{if } x(u,i) = 2 \\
            \left|\wt_{S_2}\left([b_{f(u-1)}]\right) - \wt_p\left([b_{f(u-1)}]\right)\right| & \text{if } x(u,i) = 3 \;.
        \end{cases}
    \end{eqnarray}
    Using Equation~\eqref{eq:L1-p-pdot} and since we conditioned on $E_2$ we get the desired inequality.
\end{proof}

\smallskip
We prove another claim to establish a connection between $\Dist(\widetilde{T},w',\widetilde{p})$ and $\Dist(T,w,p)$.
\begin{claim} \label{clm:final_nu_connection}
    \begin{equation}
        |2\Dist(\widetilde{T},w',\widetilde{p}) - \Dist(T,w,p)| \leq L_1(p,\Dot{p}) \;.
    \end{equation}
\end{claim}
\begin{proof}
    The claim follows by combining Equations~\eqref{eq:nu_prob}, \eqref{eq:nu_half} and~\eqref{eq:nu_split_wc}.
\end{proof}


\subsection{Wrapping things up in the general case}\label{subsubsec:wrap-general}
We can now restate and prove the main theorem of this section (as it appeared in the introduction).

\ifnum\icalpsub=1
\smallskip\noindent
 \textcolor{lipicsGray}{$\blacktriangleright$}
{\sffamily\bfseries{Theorem~\ref{thm:main}}}.
\else
\medskip\noindent
\textbf{Theorem~\ref{thm:main}}~
\fi
\emph{
\DistFreeThmStatement
}

\medskip
As in the special case of $w\in \calW_c$ and a quantized distribution $p$, the running time of the algorithm is linear in the size of the sample.
The proof of Theorem~\ref{thm:main} is similar to the proof of Lemma~\ref{lem:wrap-up-w-Wc}, but there are several important differences, and for the sake of completeness it is given in full detail.

\smallskip
\begin{proof}
    The algorithm first takes a sample $S_1$ of size $s_1 = 120 z \log(240z)$ and constructs a set of intervals $\mathcal{B}$ as defined in Definition~\ref{def:interval_construction}.
    Next the algorithm takes another sample, $S_2$, of size $s_2 = z^2 \log(40 k U)$ according to which it computes the vector $\wt_{S_2}([b_u])$ for each $u \in U$ according to Definition~\ref{def:wt-S} and defines a matrix $\breve{\xi}$ of size $k \times U$ as follows. For every $i \in [k]$ and $u \in [U]$, it sets $\breve{\xi}[i][u] = \breve{\xi}_{i}^{u}$, where $\breve{\xi}_{i}^{u}$ is as defined in Equation~\eqref{eq:xi_est_def}. Then, the algorithm defines $w' = w_1 0 w_2 0 \dots w_k 0$ and $\mathcal{B}'$, which is a set of $U'$ intervals as defined in Definition~\ref{def:interval_construction_wc}, using Algorithm \ref{alg:b_prime}, according to which it also obtains the function $f:[U']\rightarrow[U]$. Afterwards, the algorithm defines a matrix $\widehat{\xi}$ of size $2k \times U'$ as follows. For every $i \in [2k]$ and $u \in [U']$, it sets $\widehat{\xi}[i][u] = \widehat{\xi}_{i}^{u}$, where $\widehat{\xi}_{i}^{u}$ is as defined in Equation~\eqref{eq:est_nm_prime}. Lastly the algorithm outputs $\wDist = 2 \rd(\widehat{\xi})$, where $\rd$ is as defined in Definition \ref{def:nJ}.

    We would like to apply Claim \ref{clm:alg_no_n} in order to show that $|\wDist - \Dist(T,w,p)| \leq \delta$ with probability of at least $\frac{2}{3}$.
    By the setting of $s_1$, applying Claim~\ref{clm:E_1} gives us that
    with probability at least $\frac{8}{10}$, the event $E_1$, as defined in Definition~\ref{def:E1}, holds.
    By the setting of $s_2$, applying Claim~\ref{clm:E_2} gives us that with probability at least $\frac{9}{10}$ the event $E_2$, as defined in Definition~\ref{def:E2},  holds. We henceforth condition on both events (where they  hold together with probability at least $7/10$).

    In order to apply Claim~\ref{clm:alg_no_n}, we set $\widetilde{w} = w'$, $J = \left\{\widetilde{b}_0,\widetilde{b}_1,\dots,\widetilde{b}_{U'}\right\}$ (recall Definition~\ref{def:Btil_wc}) and $\widehat{\nms} = \widetilde{n}\widehat{\xi}$, for $\widehat{\xi}$ as defined above. Also we set $c_1 = \frac{1}{8}$ and $c_2 = \frac{1}{8}$. We next show that all the items in the premise of the claim are satisfied.

    To show that Item \ref{itm:J_not_spaced} is satisfied, we first note that the following is true for every $u \in [U']$. Since $\widetilde{p}$ is the uniform distribution over $[\widetilde{n}]$,
    $\frac{\widetilde{b}_{u} - \widetilde{b}_{u-1}}{\widetilde{n}} = \wt_{\widetilde{p}}(\widetilde{\mathcal{B}}_{u})$. Therefore, if $\frac{\widetilde{b}_{u} - \widetilde{b}_{u-1}}{\widetilde{n}} \ge \frac{25}{4}\frac{1}{z}$, then $\wt_{\widetilde{p}}(\widetilde{\mathcal{B}}_{u}) \ge \frac{25}{4}\frac{1}{z}$ which according to Claim \ref{clm:light_b_is_light_b} (recall we condition on $E_1$) implies that $\mathcal{B}_{f(u)}$ is heavy. This in turn means that $\mathcal{B}'_{u}$ contains only one index, which implies that $\widetilde{T}[\widetilde{b}_{u-1}+1] = \dots = \widetilde{T}[\widetilde{b}_{u}]$. By the definition of $z$ (Equation~\eqref{eq:z_eq}) and the setting of $c_1$, we get that the item is satisfied.

    To show that Item \ref{itm:good_approx} is satisfied, we use Claim \ref{clm:approx_breve_xi}, which gives us that $\left|\widehat{\xi}_{i}^{u} - \widetilde{\xi}_{i}^{u}\right| \leq \frac{c_{\eta}}{z} + \frac{1}{2z}$ for every $i \in [2k]$ and $u \in [U']$. By the setting of $c_2$ along with Equation~\eqref{eq:xi_tilde_final} and the definitions of $z$ and $c_\eta$ (the latter is defined in the beginning of Section \ref{subsec:quant-dist}), we get that the item is satisfied.

    After applying Claim \ref{clm:alg_no_n} we get that $|\wDist - 2\Dist(\widetilde{T},w')| \leq 2(c_1 \delta + 2c_2 \delta)$, which by the setting of $c_1$ and $c_2$ is at most $\frac{3 \delta}{4}$. Since $\widetilde{p}$ is the uniform distribution,  $\Dist(\widetilde{T},w') = \Dist(\widetilde{T},w',\widetilde{p})$. Using Claim \ref{clm:final_nu_connection}
    and Equation~\eqref{eq:L1-p-pdot} we get $|2\Dist(\widetilde{T},w',\widetilde{p}) - \Dist(T,w,p)| \leq 2 \frac{c_{\eta}}{z}$, which by the definition of $z$ and $c_\eta$ is at most $\frac{\delta}{4}$, so the claim follows.
\end{proof} 

\fi

%% file: lower-bound.tex
\section{A lower bound for distance approximation}\label{sec:lower-bound}
In this section we give a lower bound for the number of samples required to perform distance-approximation from $w$-freeness of a text $T$. The lower bound holds when the underlying distribution is the uniform distribution.
\begin{theorem}\label{thm:lb}
Let $\kd$ be the number of distinct symbols in $w$.
Any distance-approximation algorithm for $w$-freeness under the uniform distribution must take a sample of size
$\Omega(\frac{1}{\kd \delta^2})$, conditioned on $\delta \leq \frac{1}{300 \kd}$ and $n > \max \left\{\frac{8k}{\delta}, \frac{200}{\kd \delta^2}\right\}$.
\end{theorem}
Note that if $\delta \geq 1/\kd$, then the algorithm can simply output $0$. This is true since the number of role disjoint copies of $w$ in $T$ is at most the number of occurrences of the symbol in $w$ that is least frequent in $T$. This number is upper bounded by $\frac{n}{\kd}$, and so the distance from $w$-freeness is at most $\frac{1}{\kd}$. In this case no sampling is needed, so only the trivial lower bound holds. The proof will deal with the case of $\delta \in (0,\frac{1}{300 \kd}]$.

\smallskip
\begin{proof}
The proof is based on the difficulty of distinguishing  between an unbiased coin and a coin with a small bias. Precise details follow.

    Let $V = \left\{v_1,\dots,v_{\kd}\right\}$ be the set of distinct symbols in $w$, and let $0$ be a symbol that does not belong to $V$.
    We define two distributions over texts, $\calT_1$ and $\calT_2$ as follows.
         For each $\tau \in [\frac{n}{\kd}]$ and $\rho \in [0,1]$, let $\lambda_{\rho}^{\tau}$ be a random variable that equals $0$ with probability $\rho$ and equals $v_1$ with probability $1-\rho$.
         Let $\delta' =  3 \kd \delta$ and
    consider the following two distributions over texts
    \begin{align}
        \mathcal{T}_1 &= \left[\lambda_{\frac{1}{2}}^1,v_2,v_3,\dots,v_{\kd},\lambda_{\frac{1}{2}}^2,v_2,v_3,\dots,v_{\kd},\dots \dots,\lambda_{\frac{1}{2}}^{n/\kd},v_2,v_3,\dots,v_{\kd}\right]\;, \\
        \mathcal{T}_2 &= \left[\lambda_{\frac{1}{2}+\delta'}^1,v_2,v_3,\dots,v_{\kd},\lambda_{\frac{1}{2}+\delta'}^2,v_2,v_3,\dots,v_{\kd},\dots \dots,\lambda_{\frac{1}{2}+\delta'}^{n/\kd},v_2,v_3,\dots,v_{\kd}\right]\;.
    \end{align}
Namely, the supports of both distributions contain texts that consist of $n/\kd$ blocks of size $\kd$ each.
For $i\in \{2,\dots,\kd\}$, the $i$-th symbol in each block is $v_i$. The distributions differ only in the way the first symbol in each block is selected. In $\calT_1$ it is $0$ with probability $1/2$ and $v_1$ with probability $1/2$, while in $\calT_2$ it is $0$ with probability $1/2 + \delta' = 1/2 + 3 \delta \kd$, and $v_1$ with probability $1/2 - \delta'$.

For $b\in \{1,2\}$, consider selecting a text $T_b$ according to $\calT_b$ (denoted by $T_b \sim \calT_b$), and let $O_b$ be the number of occurrences of $v_1$ in the text (so that $O_b$ is a random variable). Observe that  $\E[O_1] = \frac{n}{2\kd}$ and
$\E[O_2] = \frac{n}{2\kd} - 3\delta n$. By applying the additive Chernoff bound (Theorem~\ref{thm:Chernoff}) and using the premise of the theorem regarding $n$,
\begin{equation}\label{eq:O1-EO1}
\Pr_{T_1\sim \calT_1}\left[ O_1 < \E[O_1] - \delta n/8 \right] \leq \exp(-2 (\kd\delta/8)^2 \cdot n/\kd) \leq \frac{1}{100}\;,
\end{equation}
and
\begin{equation}\label{eq:O1-EO2}
\Pr_{T_2\sim \calT_2} \left[ O_2 < \E[O_2] + \delta n/8 \right] \leq \exp(-2 (\kd\delta/8)^2 \cdot n/\kd) \leq \frac{1}{100}\;.
\end{equation}
%
For $b\in \{1,2\}$ let $\rdr_b = \rdr(T_b,w)$ (recall that $\rdr(T_b,w)$ denotes the number of disjoint copies of $w$ in $T_b$, and note that $\rdr_b$ is a random variable). Observe that $\rdr_1 \geq O_1 - k+1$, and $\rdr_2 \leq O_2$.

Hence, by Equation~\eqref{eq:O1-EO1}, if we select $T_1$ according to $\calT_1$ and use the premise that $n > \frac{8k}{\delta}$, then
$\rdr(T_1,w) \geq \frac{n}{2\kd} -  \frac{1}{8}\delta n - k + 1 \ge \frac{n}{2\kd} -  \frac{2}{8}\delta n$
with probability at least $99/100$, and by Equation~\eqref{eq:O1-EO2}, if we select $T_2$ according to $\calT_2$, then
$\rdr(T_2,w) \leq \frac{n}{2\kd} - 3\delta n + \frac{1}{8}\delta n  = \frac{n}{2\kd} - \frac{23}{8}\delta n$ with probability at least $99/100$.

Assume, contrary to the claim, that we have a  sample-based distance-approximation algorithm for subsequence-freeness that takes a sample of size $Q(\kd,\delta) = 1/(c\kd \delta^2)$, for some sufficiently large constant $c$, and outputs an estimate of the distance to $w$-freeness that has additive error at most $\delta$, with probability at least $2/3$.
Consider running the algorithm on either $T_1 \sim \calT_1$ or $T_2 \sim \calT_2$.
Let $L$ denote the number of times that the sample landed on an index of the form $j = \ell\cdot \kd +1$ for an integer $\ell$. By Markov's inequality, the probability that $L > 10\cdot Q(\kd,\delta)/\kd = 10/(c\kd^2 \delta^2)$ is at most $1/10$.

By the above,
if we run the algorithm on $T_1 \sim \calT_1$, then with probability at least $2/3 - 1/100 - 1/10$ the algorithm outputs an estimate $\wDist \geq \frac{n}{2 \kd} - \frac{10}{8}$  while
$L \leq 10/(c\kd^2 \delta^2)$. Similarly, if we run it on $T_2 \sim \calT_2$, then with probability at least $2/3 - 1/100 - 1/10$ the algorithm outputs an estimate $\wDist \leq \frac{n}{2 \kd} - \frac{15}{8}$  while $L \leq 10/(c\kd^2 \delta^2)$.
(In both cases the probability is taken over the selection of $T_b\sim \calT_b$, the sample that the algorithm gets, and possibly additional internal randomness of the algorithm.)
Based on the definitions of $\calT_1$ and $\calT_2$,
 this implies that it is possible to distinguish between an unbiased coin and a coin with bias $3 \kd \delta$ with probability at least $2/3 - 1/100 - 1/10 > \frac{8}{15}$, using a sample of size $\frac{1}{c' \kd^2 \delta^2}$ in contradiction to the result of Bar-Yosef~\cite[Thm. 8]{bar2003sampling} (applied with $m=2$, $\epsilon = 3 \kd \delta$. Since we have $\delta < \frac{1}{300 \kd}$, then $\epsilon < \frac{1}{96}$, as the cited theorem requires).
\end{proof}

%% file: appendix.tex
\section{Chernoff bounds} \label{app:chernoff}

\begin{theorem} \label{thm:Chernoff}
    Let $\chi_1,\dots ,\chi_m$ be $m$ independent random variables where $\chi_i \in [0,1]$ for every $1\leq i \leq m$. Let $p \eqdef \frac{1}{m}\sum_i \E[\chi_i]$. Then, for every $\gamma \in (0,1]$, the following bounds hold:
    \begin{itemize}
        \item (Additive Form)
        \begin{equation}
            \Pr\left[\frac{1}{m}\sum_{i=1}^{m} \chi_i > p + \gamma\right] < \exp\left(-2\gamma^2 m\right)
        \end{equation}
        \begin{equation}
            \Pr\left[\frac{1}{m}\sum_{i=1}^{m} \chi_i < p - \gamma\right] < \exp\left(-2\gamma^2 m\right)
        \end{equation}
        \item (Multiplicative Form)
        \begin{equation}
           \Pr\left[\frac{1}{m}\sum_{i=1}^{m} \chi_i > (1 + \gamma) p \right] < \exp\left(-\gamma^2 p m/3\right)
        \end{equation}
        \begin{equation}
            \Pr\left[\frac{1}{m}\sum_{i=1}^{m} \chi_i < (1 - \gamma) p \right] < \exp\left(-\gamma^2 p m/2\right)
        \end{equation}
    \end{itemize}
\end{theorem}

\ifnum\icalpsub=0
\section{Missing proofs}\label{app:missing}

\begin{claim}
    The greedy algorithm described as a part of the proof of Claim~\ref{clm:n[i][j]} 
    finds a maximum-size set  of role-disjoint copies of $w$ in $T$.
\end{claim}
\begin{proof}
    We start by introducing the  notion of \emph{ordered} role-disjoint copies. According to~\cite[Definition 3.8]{RR-toct}, two role-disjoint copies $C=(i_1,\dots,i_k)$ and $C'=(i'_1,\dots,i'_k)$ of $w$ in $T$ are ordered and $C'$ succeeds $C$, if $i'_j>i_j$ for every $j \in [k]$. A sequence $(C_1,\dots,C_m)$ of role-disjoint copies of $w$ in $T$ is a sequence of ordered role-disjoint copies if for every $r \in [m-1]$ we have it that $C_{r+1}$ succeeds $C_r$.

    By~\cite[Claim 3.5]{RR-toct}, for every set 
    of role-disjoint copies of $w$ in $T$, there exists a sequence 
    of ordered role-disjoint copies of $w$ in $T$ 
    with the same size.
    Since the greedy algorithm described in the proof of Claim~\ref{clm:n[i][j]} finds a sequence of ordered role-disjoint copies of $w$ in $T$, it remains to show that there is no other longer (larger) sequence of ordered role-disjoint copies of $w$ of $T$.

    Denote by $\calC = (C_1,\dots,C_{|\calC|})$ the sequence of ordered role-disjoint copies of $w$ in $T$ that is found by the greedy algorithm. 
    Assume, contrary to the claim that there is a longer sequence, $\tilcalC = (\tilC_1,\dots,\tilC_{|\tilcalC|})$ of ordered role-disjoint copies of $w$ in $T$. In what follows we show, by induction on $m$ and $i$,   that $C_m[i] \leq \tilC_m[i]$ for every pair $(m,i)\in [|\calC|]\times [k]$, which will imply a contradiction to the counter assumption.
    
    For every $m \in [|\calC|]$ and for $i=1$, by the definition of the greedy algorithm, $C_{m}[1]$ is the index of the $m$th occurrence of $w_1$ in $T$. Since  $\tilcalC$ is ordered, so that
     $\tilC_1[1] < \tilC_2[1]  < \dots < \tilC_{m}[1]$, we have that $\tilC_{m}[1]$ is the index of occurrence number $m' \geq m$ of $w_1$ in $T$. Hence $C_m[1] \leq \tilC_m[1]$ for every $m \in [|\calC|]$.
     
     In order to prove the claim for $(m,i)$ where $i> 1$, we assume by induction that it holds for $(m,i-1)$ 
     and for $(m-1,i)$, where for the sake of the argument (so that $C_{m-1}$ and $\tilC_{m-1}$ are defined also for $m=1$) we define $C_0[i] = \tilC_0[i] = -k+i$.
      By the induction hypothesis, $C_m[i-1] \leq \tilC_m[i-1]$ and 
     $C_{m-1}[i] < \tilC_{m-1}[i]$. Because indices of a copy are always strictly increasing, $\tilC_m[i-1] < \tilC_m[i]$, and since $\tilcalC$ is ordered, $\tilC_{m-1}[i] < \tilC_m[i]$. Therefore, $C_m[i-1] <  \tilC_m[i]$ and $C_{m-1}[i] < \tilC_m[i]$. By the definition of the algorithm, $C_m[i]$ is the  index of the first occurrence of $w_i$ following $C_m[i-1]$ that is larger than $C_{m-1}[i]$. Since $T[\tilC_m[i]]=w_i$ we get that $C_m[i] \leq \tilC_m[i]$, as claimed.

    Finally, by the counter assumption, $|\tilcalC| > |\calC|$. By what we have shown above, this implies that
    $C_{m}[i]  < \tilC_{|C|+1}[i]$ for every $m \in [|C|]$, and $i \in [k]$. But this contradicts the fact that the algorithm did not find any role-disjoint copy after $C_{|\calC|}$.
\end{proof}

\begin{proofof}{Claim~\ref{clm:sample}}
Let $s = \frac{\log(6k\cdot \ell)}{2\gamma^2}$. We take $s$ samples from $[n]$ selected uniformly, independently at random (allowing repetitions).
Denote the $q$-th sampled index by $\rho_q$. For every $i \in [k]$, $r \in [\ell]$ and $q \in [s]$, define the random variables $\chi_q^{i,r}$ to equal $1$ if and only if $\rho_q \in [j_r]$ and $T[\rho_q]=w_i$, Otherwise $\chi_q^{i,r} = 0$.

For every $i \in [k]$ and $r \in [\ell]$, set
\begin{equation}
    \widehat{\nms}_{i}^{r} = \frac{n}{s}\sum_{q=1}^{s}\chi_q^{i,r} \;,
\end{equation}
and notice that
\begin{equation}
    \E\left[\chi_q^{i,j}\right] = \frac{\nmr_{i}^{j_r}(T,w)}{n}\;.
\end{equation}
By the additive Chernoff bound (see Theorem~\ref{thm:Chernoff}) and the setting of $s$, we get
\begin{eqnarray}
    \Pr\left[\left|\widehat{\nms}_{i}^{r} - \nmr_{i}^{j_r}(T,w) \right| > \gamma n\right] &=& \Pr\left[\left|\frac{n}{s}\sum_{q=1}^{s}\chi_q^{i,r} - \nmr_{i}^{j_r}(T,w) \right| > \gamma n\right] \nonumber \\
    &= &\Pr\left[\left|\frac{1}{s}\sum_{q=1}^{s}\chi_q^{i,r} - \frac{\nmr_{i}^{j_r}(T,w)}{n} \right| > \gamma \right] \nonumber \\
    &<& 2 \exp\left(-2 \gamma^2 s\right) = \frac{1}{3 k\cdot \ell}\;.
\end{eqnarray}
Applying the union bound over all pairs $(i,r) \in [k]\times [\ell]$
we get that with probability at least $\frac{2}{3}$, for every $i \in [k]$ and  $r\in [\ell]$
\begin{equation}\label{eq:good-est}
    \left|\widehat{\nms}_{i}^{r}- \nmr_{i}^{j_r}(T,w) \right| \leq \gamma n \;,
\end{equation}
as required.
\end{proofof}


\medskip

\begin{proofof}{Claim~\ref{clm:nJ-n-J-prime}}
For the sake of simplicity, we use $T$ and $w$ instead of $\tilT$ and $\tilw$, respectively.
Recall that $\rd(\nms) = \rd_{k}^{\ell}(\nms)$ and $\rdr(T,w) = \rdr_{k}^{j_{\ell}}(T,w)$.
We shall prove that for every $i \in [k]$ and for every $r\in [\ell]$,
$\left|\rd_{i}^{r}(\nms) - \rdr_{i}^{j_{r}}(T,w)\right| \leq (i-1)\cdot \max_{\tau \in [r] \setminus J' }\left\{j_{\tau} - j_{\tau-1}\right\}$.
We prove this by induction on $i$.\\

\noindent
For $i=1$ and every $r\in [\ell]$:
\begin{eqnarray}
\left|\rd_{1}^r(\nms) - \rdr_{1}^{j_r}(T,w)\right| &=& \left| \nmr_1^{j_r}(T,w) -  \nmr_1^{j_r}(T,w) \right| \nonumber \\ &=&
   0  \leq (1-1) \cdot \max_{\tau \in [1] \setminus J'}\left\{j_{\tau} - j_{\tau-1}\right\} \;,
\end{eqnarray}
where the first equality follows from the setting of $\nms$ and the definitions of $\rd_1^r(\nms)$ and $\rdr_1^{j_r}(T,w)$.

\noindent
For the induction step, we assume the claim holds for $i-1 \geq 1$ (and every $r\in [\ell]$)
and prove it for $i$. We have,
\begin{eqnarray}
    \lefteqn{\rd_{i}^{r}(\nms) - \rdr_{i}^{j_r}(T,w)}\nonumber\\
     &=& \nmr_i^{j_r}(T,w)- \max_{b \in [r]}\left\{ \nmr_i^{j_b}(T,w)  -  \rd_{i-1}^{b}(\nms) \right\} - \rdr_{i}^{j_r}(T,w) \label{eq:j-to-J-1-J-prime}\\
 &=& \max_{j \in [j_r]}\left\{ \nmr_i^{j}(T,w)  -  \rdr_{i-1}^{j-1}(T,w) \right\}
     - \max_{b\in [r]}\left\{ \nmr_i^{j_b}(T,w)  -  \rd_{i-1}^{b}(\nms) \right\} \label{eq:j-to-J-2-J-prime}\;,
\end{eqnarray}
where Equation~\eqref{eq:j-to-J-1-J-prime} follows from the setting of $\nms$ and the definition of $\rd_{i}^{r}(\nms)$,
and Equation~\eqref{eq:j-to-J-2-J-prime} is implied by Claim~\ref{clm:n[i][j]}. Denote by $j^*$ an index $j \in [j_r]$ that maximizes the first max term and
let $b^*$ be the smallest index such that $j_{b^*} \ge j^*$.
We have:
\begin{eqnarray}
    \lefteqn{\max_{ j\in[j_r]}\left\{ \nmr_i^{j}(T,w)  -  \rdr_{i-1}^{j-1}(T,w) \right\} - \max_{b \in [r]}\left\{ \nmr_i^{j_b}(T,w)  -  \rd_{i-1}^{b}(\nms) \right\}} \nonumber \\
    &\leq&  \nmr_i^{j^*}(T,w)  -  \rdr_{i-1}^{j^*-1}(T,w) - \nmr_i^{j_{b^*}}(T,w) + \rd_{i-1}^{b^*}(\nms) \nonumber\\
    & =&  \nmr_i^{j^*}(T,w) + \rdr_{i-1}^{j_{b^*}}(T,w) - \rdr_{i-1}^{j_{b^*}}(T,w)  -  \rdr_{i-1}^{j^*-1}(T,w) \nonumber \\ &&- \nmr_i^{j_{b^*}}(T.w) + \rd_{i-1}^{b^*}(\nms) \nonumber \\
    &\leq& \left(\rd_{i-1}^{b^*}(\nms) -  \rdr_{i-1}^{j_{b^*}}(T,w)\right) +  \left(\nmr_i^{j^*}(T,w) - \nmr_i^{j_{b^*}}(T,w) \right) \nonumber \\ &&+\left(\rdr_{i-1}^{j_{b^*}}(T,w) - \rdr_{i-1}^{j^*-1}(T,w)\right)
    \nonumber\\
    &\leq& (i-2) \max_{\tau \in [r] \setminus J' }\left\{j_{\tau} - j_{\tau-1}\right\} \nonumber \\ &&+
    \begin{cases}
        0 & \text{, if } T[j_a^*] = \dots = T[j_b^*]] \\
        \max_{\tau \in [r] \setminus J' }\left\{j_{\tau} - j_{\tau-1}\right\} & \text{otherwise}\\
    \end{cases} \nonumber\\
    &\leq& (i-1)\max_{\tau \in [r] \setminus J' }\left\{j_{\tau} - j_{\tau-1}\right\}\label{eq:max-j-max-b-J-prime}\;,
\end{eqnarray}
Where in the third inequality we used the induction assumption and the fact that if we don't have $T[j_{b^*}] = \dots = T[j^*]]$, then $\left(\nmr_i^{j^*}(T,w) - \nmr_i^{j_{b^*}}(T,w) \right) +\left(\rdr_{i-1}^{j_{b^*}}(T,w) - \rdr_{i-1}^{j^{*}-1}(T,w)\right) \leq (j_{b^*} - j^* + 1)\\ \leq \max_{\tau \in [r] \setminus J' }\left\{j_{\tau} - j_{\tau-1}\right\}$.\\

By combining Equations~\eqref{eq:j-to-J-2-J-prime} and~\eqref{eq:max-j-max-b-J-prime}, we get that
\begin{equation}
\rd_{i}^{r}(\nms) - \rdr_{i}^{j_r}(T,w) \leq (i-1)\max_{\tau \in [r] \setminus J' }\left\{j_{\tau} - j_{\tau-1}\right\}\;. \label{eq:nms-T-w-J-prime}
\end{equation}
Similarly to Equation~\eqref{eq:j-to-J-2-J-prime}:
\begin{equation}
    \rdr_{i}^{j_r}(T,w) - \rd_{i}^{r}(\nms)
     = \max_{b\in [r]}\left\{ \nmr_i^{j_b}(T,w)  -  \rd_{i-1}^{b}(\nms) \right\} - \max_{j\in [j_r]}\left\{ \nmr_i^{j}(T,w)  -  \rdr_{i-1}^{j-1}(T,w) \right\}\;.
     \label{eq:J-2-to-j-J-prime}
\end{equation}
Let $b^{**}$ be the index $b\in [r]$ that maximizes the first max term.
We have:
\begin{eqnarray}
    \lefteqn{\max_{b \in [r]}\left\{ \nmr_i^{j_b}(T,w)  -  \rd_{i-1}^{b}(\nms) \right\} - \max_{ j \in [j_r]}\left\{ \nmr_i^{j}(T,w)  -  \rdr_{i-1}^{j-1}(T,w) \right\}} \nonumber \\
    & \leq& \nmr_i^{j_{b^{**}}}(T,w)  -  \rd_{i-1}^{{b^{**}}}(\nms) - \nmr_i^{j_{b^{**}}}(T,w) + \rdr_{i-1}^{j_{b^{**}}-1}(T,w) \nonumber\\
    & \leq&  \rdr_{i-1}^{j_{b^{**}}}(T,w)  -  \rd_{i-1}^{b^{**}}(\nms) \nonumber\\
    & \leq&  \left|\rdr_{i-1}^{j_{b^{**}}}(T,w)  -  \rd_{i-1}^{b^{**}}(\nms)\right| \nonumber\\
    & \leq& (i-2)\max_{\tau \in [r] \setminus J' }\left\{j_{\tau} - j_{\tau-1}\right\}
    \leq (i-1)\max_{\tau \in [r] \setminus J' }\left\{j_{\tau} - j_{\tau-1}\right\} \;.
     \label{eq:max-b-max-j-J-prime}
\end{eqnarray}
Hence (combining Equations~\eqref{eq:J-2-to-j-J-prime} and~\eqref{eq:max-b-max-j-J-prime}),\footnote{It actually holds that $\rd_{i}^{r}(\nms) \geq \rdr_{i}^{j_r}(T,w)$, so that $\rdr_{i}^{j_r}(T,w) - \rd_{i}^{r}(\nms)\leq 0$, but for the sake of simplicity of the inductive argument, we prove the same upper bound on $\rdr_{i}^{j_r}(T,w) - \rd_{i}^{r}(\nms)$ as on $\rd_{i}^{r}(\nms) - \rdr_{i}^{j_r}(T,w)$.}
\begin{equation}
     \rdr_{i}^{j_r}(T,w) - \rd_{i}^{r}(\nms) \leq (i-1)\max_{\tau \in [r] \setminus J' }\left\{j_{\tau} - j_{\tau-1}\right\} \label{eq:T-w-nms-J-prime} \;
\end{equation}
Together, Equations~\eqref{eq:nms-T-w-J-prime} and~\eqref{eq:T-w-nms-J-prime} give us that
\begin{equation}
    \left|\rd_{i}^{r}(\nms) - \rdr_{i}^{j_r}(T,w)\right| \leq (i-1)\max_{\tau \in [r] \setminus J' }\left\{j_{\tau} - j_{\tau-1}\right\}\;,
\end{equation}
and the proof is completed.
\end{proofof}

\fi

%% file: main.bbl
\begin{thebibliography}{10}

\bibitem{ACCL}
Nir Ailon, Bernard Chazelle, Seshadhri Comandur, and Ding Liu.
\newblock Estimating the distance to a monotone function.
\newblock {\em Random Structures and Algorithms}, 31(3):371--383, 2007.

\bibitem{bar2003sampling}
Ziv {Bar-Yossef}.
\newblock Sampling lower bounds via information theory.
\newblock In {\em Proceedings of the 35th Annual ACM Symposium on the Theory of
  Computing}, pages 335--344, 2003.

\bibitem{BFLR}
Omri Ben{-}Eliezer, Eldar Fischer, Amit Levi, and Ron~D. Rothblum.
\newblock Hard properties with (very) short {PCPP}s and their applications.
\newblock In {\em Proceedings of the 11th Innovations in Theoretical Computer
  Science conference (ITCS)}, pages 9:1--9:27, 2020.

\bibitem{berman2016tolerant}
Piotr Berman, Meiram Murzabulatov, and Sofya Raskhodnikova.
\newblock Tolerant testers of image properties.
\newblock {\em ACM Transactions on Algorithms}, 18(4):1--39, 2022.
\newblock Article number 37.

\bibitem{berman2014lp}
Piotr Berman, Sofya Raskhodnikova, and Grigory Yaroslavtsev.
\newblock Lp-testing.
\newblock In {\em Proceedings of the 46th Annual ACM Symposium on the Theory of
  Computing}, pages 164--173, 2014.

\bibitem{black2020domain}
Hadley Black, Deeparnab Chakrabarty, and C.~Seshadhri.
\newblock Domain reduction for monotonicity testing: {A} $o(d)$ tester for
  boolean functions in $d$-dimensions.
\newblock In {\em Proceedings of the 31st Annual ACM-SIAM Symposium on Discrete
  Algorithms}, pages 1975--1994, 2020.

\bibitem{blais2019tolerant}
Eric Blais, Cl{\'e}ment~L Canonne, Talya Eden, Amit Levi, and Dana Ron.
\newblock Tolerant junta testing and the connection to submodular optimization
  and function isomorphism.
\newblock {\em ACM Transactions on Computation Theory}, 11(4):1--33, 2019.

\bibitem{BFH}
Eric Blais, Renato {Ferreira~Pinto~Jr.}, and Nathaniel Harms.
\newblock {VC} dimension and distribution-free sample-based testing.
\newblock In {\em Proceedings of the 53rd Annual ACM Symposium on the Theory of
  Computing}, pages 504--517, 2021.

\bibitem{BH}
Avrim Blum and Lunjia Hu.
\newblock Active tolerant testing.
\newblock In {\em Proceedings of the 31st Conference on Computational Learning
  Theory (COLT)}, pages 474--497, 2018.

\bibitem{braverman2022improved}
Mark Braverman, Subhash Khot, Guy Kindler, and Dor Minzer.
\newblock Improved monotonicity testers via hypercube embeddings.
\newblock In {\em Proceedings of the 13th Innovations in Theoretical Computer
  Science conference (ITCS)}, pages 25:1--25:24, 2024.

\bibitem{CGR13}
Andrea Campagna, Alan Guo, and Ronitt Rubinfeld.
\newblock Local reconstructors and tolerant testers for connectivity and
  diameter.
\newblock In {\em Proceedings of the 17th International Workshop on
  Randomization and Computation}, pages 411--424, 2013.

\bibitem{canonne2019testing}
Cl{\'e}ment~L Canonne, Elena Grigorescu, Siyao Guo, Akash Kumar, and Karl
  Wimmer.
\newblock Testing $k$-monotonicity: The rise and fall of boolean functions.
\newblock {\em Theory of Computing}, 15(1):1--55, 2019.
\newblock This paper appeared in the proceedings of ITCS 2017.

\bibitem{CS23}
Omer {Cohen Sidon}.
\newblock Sample-based distance-approximation for subsequence-freeness, 2023.
\newblock MSc thesis, Tel Aviv University.

\bibitem{DK16}
Ilias Diakonikolas and Daniel Kane.
\newblock A new approach for testing properties of discrete distributions.
\newblock In {\em Proceedings of the 56th Annual IEEE Symposium on Foundations
  of Computer Science}, pages 685--694, 2016.

\bibitem{FR}
Shahar Fattal and Dana Ron.
\newblock Approximating the distance to monotonicity in high dimensions.
\newblock {\em ACM Transactions on Algorithms}, 6(3):1--37, 2010.

\bibitem{fiat2021efficient}
Nimrod Fiat and Dana Ron.
\newblock On efficient distance approximation for graph properties.
\newblock In {\em Proceedings of the 32nd Annual ACM-SIAM Symposium on Discrete
  Algorithms}, pages 1618--1637, 2021.

\bibitem{FF}
Eldar Fischer and Lance Fortnow.
\newblock Tolerant versus intolerant testing for boolean properties.
\newblock {\em Theory of Computing}, 2:173--183, 2006.

\bibitem{FN}
Eldar Fischer and Ilan Newman.
\newblock Testing versus estimation of graph properties.
\newblock {\em SIAM Journal on Computing}, 37(2):482--501, 2007.

\bibitem{GGR}
Oded Goldreich, Shafi Goldwasser, and Dana Ron.
\newblock Property testing and its connections to learning and approximation.
\newblock {\em Journal of the ACM}, 45:653--750, 1998.

\bibitem{GuRu}
Venkat Guruswami and Atri Rudra.
\newblock Tolerant locally testable codes.
\newblock In {\em Proceedings of the 9th International Workshop on
  Randomization and Computation}, pages 306--317, 2005.

\bibitem{harms2022downsampling}
Nathaniel Harms and Yuichi Yoshida.
\newblock Downsampling for testing and learning in product distributions, 2022.

\bibitem{hoppen2017estimating}
Carlos Hoppen, Yoshiharu Kohayakawa, Richard Lang, Hanno Lefmann, and Henrique
  Stagni.
\newblock Estimating the distance to a hereditary graph property.
\newblock {\em Electronic Notes in Discrete Mathematics}, 61:607--613, 2017.

\bibitem{kopparty2009tolerant}
Swastik Kopparty and Shubhangi Saraf.
\newblock Tolerant linearity testing and locally testable codes.
\newblock In {\em Proceedings of the 13th International Workshop on
  Randomization and Computation}, pages 601--614. 2009.

\bibitem{levi2018lower}
Amit Levi and Erik Waingarten.
\newblock Lower bounds for tolerant junta and unateness testing via rejection
  sampling of graphs.
\newblock In {\em Proceedings of the 10th Innovations in Theoretical Computer
  Science conference (ITCS)}, pages 52:1--52:20, 2019.

\bibitem{MR}
Sharon Marko and Dana Ron.
\newblock Distance approximation in bounded-degree and general sparse graphs.
\newblock {\em Transactions on Algorithms}, 5(2), 2009.
\newblock Article number 22.

\bibitem{NV-LIS}
Ilan Newman and Nithin Varma.
\newblock New sublinear algorithms and lower bounds for {LIS} estimation.
\newblock In {\em Automata, Languages and Programming: 48th International
  Colloquium}, pages 100:1--100:20, 2021.

\bibitem{pallavoor2022approximating}
Ramesh Krishnan~S Pallavoor, Sofya Raskhodnikova, and Erik Waingarten.
\newblock Approximating the distance to monotonicity of boolean functions.
\newblock {\em Random Structures \& Algorithms}, 60(2):233--260, 2022.

\bibitem{PRR}
Michal Parnas, Dana Ron, and Ronitt Rubinfeld.
\newblock Tolerant property testing and distance approximation.
\newblock {\em Journal of Computer and System Sciences}, 72(6):1012--1042,
  2006.

\bibitem{RR-toct}
Dana Ron and Asaf Rosin.
\newblock Optimal distribution-free sample-based testing of
  subsequence-freeness with one-sided error.
\newblock {\em ACM Transactions on Computation Theory}, 14(4):1--31, 2022.
\newblock An extended abstract of this work appeared in the proceedings of SODA
  2021.

\bibitem{RS}
Ronitt Rubinfeld and Madhu Sudan.
\newblock Robust characterization of polynomials with applications to program
  testing.
\newblock {\em SIAM Journal on Computing}, 25(2):252--271, 1996.

\end{thebibliography}
